\documentclass[final]{dmtcs-episciences}

\usepackage[utf8]{inputenc}
\usepackage{amsmath}
\usepackage{tikz}
\usepackage{graphicx}
\usepackage{subcaption}
\usepackage[normalem]{ulem}
\usepackage{lscape}
\usepackage{diagbox}
\usepackage{float}

\usetikzlibrary{shapes.geometric, arrows}

\numberwithin{equation}{section}

\newtheorem{thm}{Theorem}[section]
\newtheorem{defn}[thm]{Definition}
\newtheorem{observation}[thm]{Observation}
\newtheorem{lem}[thm]{Lemma}
\newtheorem{cor}[thm]{Corollary}
\newtheorem{conjecture}[thm]{Conjecture}

\begin{document}

\author{Sa\'ul A. Blanco\affiliationmark{1} \and Charles Buehrle\affiliationmark{2} \and Akshay Patidar\affiliationmark{3}\thanks{Supported by the Global Talent Attraction Program (GTAP) administered by the School of Informatics, Computing, and Engineering at Indiana University.}}%\corref{cor}}

\title[Stacks requiring four flips to be sorted]{On the number of pancake stacks requiring four flips to be sorted}

\affiliation{Department of Computer Science, Indiana University\\
Department of Mathematics, Physics, and Computer Studies, Notre Dame of Maryland University\\
Department of Computer Science and Engineering, Indian Institute of Technology}


\received{2019-2-27}

\revised{2019-8-6}

\accepted{2019-10-15}

\publicationdetails{21}{2019}{2}{5}{5214}
\maketitle 
\begin{abstract}
    Using existing classification results for the 7- and 8-cycles in the pancake graph, we determine the number of permutations that require 4 pancake flips (prefix reversals) to be sorted.  A similar characterization of the 8-cycles in the burnt pancake graph, due to the authors, is used to derive a formula for the number of signed permutations requiring 4 (burnt) pancake flips to be sorted. We furthermore provide an analogous characterization of the 9-cycles in the burnt pancake graph. Finally we present numerical evidence that polynomial formulas exist giving the number of signed permutations that require $k$ flips to be sorted, with $5\leq k\leq9$. 
\end{abstract}

\section{Introduction}
The idea of sorting permutations is a classical one in combinatorics and computer science. Of particular interest is when the sorting is performed utilizing only prefix reversals (see, for example,~\cite{Chitt,Cibulka, GatesPapa}). Problems like the \emph{pancake problem} that ask to determine the minimum number of prefix-reversal flips that are needed to sort any permutation in $S_n$ are computationally hard to solve~\cite{BulFerRusu}. In this context a natural question arises: How many permutations in $S_n$ require $k$ prefix-reversal flips to be sorted?  One of our main contributions is the answer to the case $k=4$ (the cases $k\leq 3$ are trivial) for both permutations and signed permutations. Our methods rely on an existing classification of the 6-,7-, and 8-cycles in the pancake graph and a classification of the 8-cycles in the burnt pancake graph, which is due to the authors~\cite{BBP2018}. In particular, our main results are the following.
\begin{enumerate}
\item[(I)]\label{main:one} We provide an explicit formula that gives the number of permutations in $S_n$ that require exactly four prefix reversal flips to be sorted. This formula is given by a simple integer-valued polynomial
\begin{displaymath}\frac{1}{2} \left(2 n^4 - 15 n^3 + 29 n^2 + 6 n - 34\right),\end{displaymath} where $n\geq4$. The details are given in Section~\ref{sec:sn}.
\item[(II)] We give an explicit formula for the number of signed permutations in $B_n$ that require exactly four prefix-reversal flips to be sorted. This formula is also given by a simple integer-valued polynomial, namely,
\[
 \frac{1}{2} n(n-1)^2(2n-3),
\] 
with $n\geq1$. The details are given in Section~\ref{sec:bn}.
\item[(III)] We also provide a classification of all the 9-cycles in the burnt pancake graph. Concretely, we prove that all of these 9-cycles can be described by two canonical forms. The details are given in Section~\ref{sec:9cycles}.
\end{enumerate} 

We point out that the polynomial in (I) can also be derived from the algorithm described in~\cite{HV16}, where the authors utilize structural properties of certain permutations to obtain generating functions. Our methods are entirely elementary and rely on the classification of cycles inside the pancake and burnt pancake graph and on the principle of inclusion-exclusion. Presently, the algorithm in~\cite{HV16} cannot be applied to signed permutations, and in particular, no other proof for our second main result (II) is known. 

The reader is referred to Section~\ref{sec:preliminaries} for the basic definitions and notation used throughout the paper. We end the paper with conjectures for closed formulas giving the number of signed permutations requiring exactly $k$ prefix-reversal flips to be sorted with $5\leq k\leq 9$, and other conjectures.  

\section{Preliminaries and notation}\label{sec:preliminaries}
Throughout this note, $n$ will denote a positive integer greater than 1, and for a positive integer $k\leq n$, $[k]$ will denote the set $\{1,2,\ldots,k\}$. We will use $S_n$ to denote the set of permutations of the set $[n]$. Furthermore, if $2\leq i\leq n$, we denote by $r_i$ the following permutation, written in one-line notation:
\[
r_i= i\,(i-1)\,\cdots\,2\,1\,(i+1)\,\cdots\,n.
\]

The elements of the set $R:=\{r_i\}_{i=2}^n$ are referred to as \textit{prefix reversals} or \textit{pancake flips}. 

Similarly one can define prefix reversals on the group of signed permutations. A \textit{signed permutation} is a permutation $w$ of the set $[\pm n]:=\{-n,-(n-1),\ldots,-2,-1,1,2,\ldots n-1,n\}$ satisfying $w(-i)=-w(i)$ for all $i\in[n]$. For convenience, we will use $\underline{i}$ instead of $-i$. We will use \textit{window notation} (see~\cite[Section 8.1]{BjornerBrenti}) to denote signed permutations. More specifically, we will use $[w(1)\,w(2)\,\ldots\,w(n)]$ to denote $w$. In Coxeter groups literature, the group of signed permutations is denoted by $B_n$ (see~\cite[Chapter 8]{BjornerBrenti}). In this context, if $i\in [n]$, a \text{signed prefix reversal} is given by
\[
r_i^{B}=[\underline{i}\,\underline{i-1}\,\cdots\,\underline{1}\,(i+1)\,(i+2)\,\cdots\,n].
\]

We refer to the elements of the set $R^B:=\{r_i^B\}_{i=1}^n$ as \textit{signed prefix reversals} or \textit{burnt pancake flips}. If clear from the context, we may drop the words ``signed," ``burnt," and the $B$ superscript. It is worth pointing out that for any $i\in[n]$, $r_i$ and $r_i^B$ are both involutions (elements of order two) of $S_n$ and $B_n$, respectively. That is, 
\[ \left(r_i\right)^2 = 1\, 2\, \cdots\, (i-1)\, i\, (i+1) \cdots n\text{, and}\]
\[ \left(r_i^B\right)^2 = [1\, 2\, \cdots\, (i-1)\, i\, (i+1) \cdots n]. \]

The graph $P_n:=(S_n,E_n)$ where
\[
E_n:= \{\{\pi,\pi r_i\}: \pi\in S_n,2\leq i\leq n\}
\] is called the \emph{pancake graph} of order $n$. Similarly, the graph $BP_n:=(B_n,E_n^B)$, where 
\[
E^B_n:= \{\{\pi,\pi r^B_i\}: \pi\in B_n,i\in[n]\},
\] is called the \textit{burnt pancake graph}. 

Both $P_n$ and $BP_n$ are Cayley graphs of $S_n$ and $B_n$, respectively, since the groups $S_n$ and $B_n$ are generated by $R$ and $R^B$, respectively. Therefore, both $P_n$ and $BP_n$ are \textit{vertex transitive} graphs; that is, given any two vertices $u,v$ in a vertex transitive graph, there exists a graph isomorphism $f$ such that $v=f(u)$.

We will use $P_{n-1}(q)$ ($BP_{n-1}(q)$, respectively) to denote the subgraph of $P_n$ ($BP_n$, respectively) induced by the subset of $S_n$ ($B_n$, respectively) of all permutations that end with $q$, with $q\in[n]$ ($[\pm n]$, respectively) in one-line notation. Furthermore, for $1 < k < n$, we use $P_{k-1}(p)$ ($BP_{k-1}(p)$, respectively) to denote the subgraph of $P_{n-1}(n)$ ($BP_{n-1}(n)$, respectively) whose vertices are the set of all $\pi \in S_{n}$ with \begin{displaymath}\pi = \pi_1\, \pi_2\, \cdots \,\pi_{k-1}\, p \,(k+1)\, (k+2)\, \cdots\, n,\end{displaymath} where $p\in[k]$ and $\pi_i \in [k] \setminus \{p\}$ or, respectively, $\pi \in B_{n}$ with \begin{displaymath}\pi = [\pi_1\, \pi_2\, \cdots \,\pi_{k-1}\, p\, (k+1)\, (k+2)\, \cdots \,n],\end{displaymath} where $p \in [\pm k]$ and $\pi_i \in [\pm k] \setminus \{p\}$. The edges of $P_{k}(p)$ being $\{\{\pi,\pi r_i\}: \pi \in P_{k}(p), 2\leq i \leq k-1\}$ ($\{\{\pi,\pi r_i\}: \pi \in BP_{k}(p), i \in [k-1]\}$, respectively). One readily notices that each $P_k(p)$ ($BP_k(p)$, respectively) is isomorphic to $P_{k-1}$ ($BP_{k-1}$, respectively).

A key result that we will use is the following classification of the 6-, 7-, and 8-cycles in $P_n,n\geq4$. We will refer to each cycle $C$ in $P_n$ ($BP_n$, respectively) by listing the edges that form $C$ consecutively, and we label each edge $\{\pi,\pi r_i\}$ ($\{\pi,\pi r^B_i\}$, respectively), with $\pi\in S_n,r_i\in R$ (with $\pi\in B_n,r_i^B\in R^B$, respectively) by $r_i$ ($r_i^B$, respectively). Since there are multiple ways to refer to a cycle, we choose a canonical form for every cycle. We say that a cycle $C$ is in \emph{canonical form} if $C=r_{i_1}\cdots r_{i_{\ell}}$ ($C=r^B_{i_1}\cdots r^B_{i_{\ell}}$, if in $BP_n$) and the sequence $(i_1,\ldots,i_{\ell})$ is lexicographically maximal among all sequences corresponding to indices of prefix reversals that would also traverse $C$. For example, $r_3 r_2 r_3 r_2 r_3 r_2$ is in canonical form whereas $r_2 r_3 r_2 r_3 r_2 r_3$ is not since $(232323)\underset{\text{lex}}{<}(323232)$. We are now ready to spell out the cycle classifications that are used in the proofs of our main results.

The following is an amalgam of results spanning three articles written by Konstantinova and Medvedev, which classifying the canonical forms of small cycles (6-,7-,8-, and 9-cycles) in $P_n$. This single theorem is actually a restatement of four separate results found in \cite[Lemma 3]{KM10}, \cite[Theorem 1]{KM10}, \cite[Theorem 1.3]{KonMed}, and \cite[Theorem 4]{KM11}.
\begin{thm}\label{t:cycleclassification} If $n\geq3$, then

there is only one canonical 6-cycle in $P_n$:
\begin{align}
    \label{P6-1}&r_3 r_2 r_3 r_2 r_3 r_2.
\end{align}
Furthermore, the 7-cycles in $P_n$ have the following canonical forms:
\begin{align}
    \label{P7-1}&r_k r_{k-1} r_k r_{k-1} r_{k-2} r_k r_2, &\quad 4 \leq k \leq n.
\end{align}
Moreover, the 8-cycles in $P_n$ have the following canonical forms: 
\begin{align}
    \label{P8-1}&r_{k}r_{j}r_{i}r_{j}r_{k}r_{k-j+i}r_{i}r_{k-j+i}, &\quad 2 \leq i < j \leq k-1, 4 \leq k \leq n,\\
    \label{P8-2}&r_{k}r_{k-1}r_{2}r_{k-1}r_{k}r_{2}r_{3}r_{2}, &\quad 4 \leq k \leq  n,\\
    \label{P8-3}&r_{k}r_{k-i}r_{k-1}r_{i}r_{k}r_{k-i}r_{k-1}r_{i}, &\quad 2 \leq i \leq k-2, 4 \leq k \leq n,\\
    \label{P8-4}&r_{k}r_{k-i+1}r_{k}r_{i}r_{k}r_{k-i}r_{k-1}r_{i-1}, &\quad 3 \leq i \leq k-2, 5 \leq k \leq n,\\
    \label{P8-5}&r_{k}r_{k-1}r_{i-1}r_{k}r_{k-i+1}r_{k-i}r_{k}r_{i}, &\quad 3 \leq i \leq k-2, 5 \leq k \leq n,\\
    \label{P8-6}&r_{k}r_{k-1}r_{k}r_{k-i}r_{k-i-1}r_{k}r_{i}r_{i+1}, &\quad 2 \leq i \leq k-3, 5 \leq k \leq n,\\
    \label{P8-7}&r_{k}r_{k-j+1}r_{k}r_{i}r_{k}r_{k-j+1}r_{k}r_{i}, &\quad 2 \leq i < j \leq k-1, 4 \leq k \leq n,\text{ and}\\
    \label{P8-8}&r_4r_3r_4r_3r_4r_3r_4r_3.
\end{align}
Additionally, the 9-cycles in $P_n$ have the following canonical forms:
\begin{align}
    \label{P9-1}&r_{k}r_{k-1}r_{i}r_{k-1}r_{k}r_{i}r_{i-1}r_{i+1}r_{2}, &\quad 3 \leq i \leq k-2, 5 \leq k \leq n,\\
    \label{P9-2}&r_{2}r_{k-i+2}r_{k}r_{i-2}r_{i-1}r_{i}r_{i-1}r_{k}r_{k-i+2}, &\quad 4 \leq i \leq k-1, 5 \leq k \leq n,\\
    \label{P9-3}&r_{k}r_{k-i}r_{k-1}r_{k-j+i-1}r_{k-j}r_{k}r_{j-i+1}r_{j}r_{i}, &\quad 2 \leq i < j \leq k-2, 5 \leq k \leq n,\\
    \label{P9-4}&r_{k}r_{k-1}r_{i}r_{i-1}r_{k-1}r_{k}r_{i}r_{i+1}r_{2}, &\quad 3 \leq i \leq k-2, 5 \leq k \leq n,\\
    \label{P9-5}&r_{k}r_{k-1}r_{k-2}r_{k-1}r_{k-2}r_{k}r_{3}r_{k}r_{k-2}, &\quad 4 \leq k \leq n,\\
    \label{P9-6}&r_{k}r_{k-1}r_{k-2}r_{i}r_{k}r_{2}r_{k}r_{i}r_{k-1}, &\quad 2 \leq i \leq k-3, 5 \leq k \leq n,\\
    \label{P9-7}&r_{k}r_{k-j+i}r_{k}r_{j}r_{i}r_{k}r_{k-j}r_{k-i}r_{j-i}, &\quad 2 \leq i \leq j-2, i+2 \leq j \leq k-2, 6 \leq k \leq n,\\
    \label{P9-8}&r_{k}r_{k-j+i}r_{k-j}r_{k}r_{j}r_{i}r_{k}r_{k-i}r_{j-i}, &\quad 2 \leq i \leq j-2, i+2 \leq j \leq k-2, 6 \leq k \leq n,\\
    \label{P9-9}&r_{k}r_{k-j+i}r_{k-j+1}r_{k}r_{j}r_{i}r_{k}r_{k-i+1}r_{j-i+1}, &\quad 2 \leq i < j \leq k-1, 4 \leq k \leq n,\\
    \label{P9-10}&r_{k}r_{k-1}r_{k}r_{k-1}r_{k}r_{k-1}r_{k-3}r_{k}r_{3}, &\quad 5 \leq k \leq n.
\end{align}
\end{thm}

In the same spirit as the classification of the cycles in $P_n$, the authors proved the following theorem classifying the 8-cycles in $BP_n$, for $n\geq2$.

\begin{thm}[Theorem 4.1 in \cite{BBP2018}]\label{t:burntpancakeclassification}
If $n\geq2$ then each 8-cycle in $BP_n$ has one of the following canonical forms:

\begin{align}
    \label{BP-8-1}
    &r_{k}r_{j}r_{i}r_{j}r_{k}r_{k-j+i}r_{i}r_{k-j+i}, & 1 \leq i < j \leq k-1, 3 \leq k \leq n,\\
    \label{BP-8-2}
    &r_{k}r_{j}r_{k}r_{i}r_{k}r_{j}r_{k}r_{i}, & 2 \leq i,j \leq k-2, i+j \leq k, 4 \leq k \leq n,\\
    \label{BP-8-3}
    &r_{k}r_{i}r_{k}r_{1}r_{k}r_{i}r_{k}r_{1}, &2 \leq i \leq k-1,  3 \leq k \leq n,\text{ and}\\
    \label{BP-8-4}
    &r_{k}r_{1}r_{k}r_{1}r_{k}r_{1}r_{k}r_{1} & 2 \leq k \leq n.
\end{align}
\end{thm}
%%%%%%%%%%%%

Our question is a very natural one: How many pancake stacks require $k$ flips to be sorted? Equivalently, how many permutations in $S_n$ require composition with $k$ prefix reversals to be sorted? Naturally, we will think of permutations instead of pancake stacks for convenience in notation. In this light, we define the \emph{pancake distance} between two permutations as follows.

\begin{defn}
Given $\pi_1,\pi_2\in S_n$, we write $d(\pi_1,\pi_2)=k$ if $\pi_2=\pi_1 r_{i_1}\cdots r_{i_k}$ for some $r_{i_1},\ldots,r_{i_k}\in R$, and $k$ is minimal. Namely, if $\pi_2=\pi_1 r_{j_1}r_{j_2}\cdots r_{j_{k'}}$ then $k\leq k'$. We call $d(\cdot,\cdot)$, the pancake distance. 

By the same token, if given $\sigma_1,\sigma_2\in B_n$ such that $\sigma_1=\sigma_2r^B_{i_1}r^B_{i_2}\cdots r^B_{i_{\ell}}$ and $\ell$ is minimal, we say that the \textit{burnt pancake distance} between $\sigma_1$ and $\sigma_2$ is $\ell$ and write $d^B(\sigma_1,\sigma_2)=\ell$.
\end{defn}

If $k\geq1$, we denote by $R_k(n)$ the number of permutations in $S_n$ that require $k$ flips to be sorted, that is, $R_k(n)=|\{\pi\in S_n: d(e,\pi)=k\}|$. Similarly, we use $R_k^B(n)$ to denote the cardinality of the set $|\{\pi\in B_n: d^B(e,\pi)=k\}|$

One easily sees that $R_0(n)=1$ (corresponding to the identity permutation that is already sorted) and that $R_1(n)=n-1$, as the only stacks that can be sorted with one flip are the prefix-reversal permutations. For $k=2,3$ the cycle structure of $P_n$ allows us to conclude that 
\[
R_2(n)=(n-1)(n-2)\text{ and }R_3(n)=(n-1)(n-2)^2-1\text { if }n\geq3.
\] 
Indeed, the smallest cycle in $P_n$ is a 6-cycle and \textit{there is only one} such a cycle, namely, $r_3 r_2 r_3 r_2 r_3 r_2$ (see~\cite[Theorem 1.1-1.2]{KonMed}). Hence, since there are $n-1$ prefix reversals, $R_2(n)=(n-1)(n-2)$ and $R_3(n)=(n-1)(n-2)^2-1$. The first non-trivial computation is $R_4(n)$. 

We have computed values of $R_k(n)$ and $R^B_k(n)$ for several instances of $n,k$ utilizing a system with Dual Xeon CPUs and 256GB of RAM for the largest computations. We have summarized the values found in Table~\ref{tab:sn} and Table~\ref{tab:bn} and remark that due to computing limitations, several entries are still unknown. However, we were able to compute enough values to be able to offer some conjectures, which we present in Section~\ref{s:conclusion}. We are happy to share our code upon request. 

We are now ready to prove the first main result of this paper, an explicit description of $R_4(n)$. 

\section{Permutations requiring four flips to be sorted}\label{sec:sn}

The approach we use to obtain the number of pancake stacks that require four flips is the following: we will use the principle of inclusion-exclusion (PIE) using a family of sets $A_i$ and we will use the classification of 7- and 8-cycles in $P_n$ when obtaining the cardinality of the intersections of said sets. More formally, our aim is to obtain the cardinality of the set $A^4:=\{\pi\in S_n: d(e,\pi)=4\}.$ We furthermore let $A_i\subseteq A^4$, for $0\leq i\leq 4$, be the following sets.

\begin{enumerate}
    \item If $0\leq i\leq 3$,  \begin{displaymath}A_i := \{\pi=r_{j_1}r_{j_2}r_{j_3}r_{j_4} \in A^4: j_{i+1}=n,j_k\neq n\text{ for }k<i+1\}.\end{displaymath}
    \item If $i=4$, \begin{displaymath}A_4 := \{\pi=r_{j_1}r_{j_2}r_{j_3}r_{j_4} \in A^4:  j_k\neq n \text{ for all }k\in[4]\}.\end{displaymath}
\end{enumerate}

In other words, if one were to think of a path between $e$ and $\pi\in A^4$ in $P_n$, $A_i$ would contain all the paths that have $i+1$ vertices inside $P_{n-1}(n)$. It follows that 
\begin{equation}\label{eq:PIE}
R_4(n)=|A^4|=\left|\bigcup_{i=0}^n A_i\right|.
\end{equation}
We will utilize PIE to compute the cardinality of the union $\bigcup_{i=0}^n A_i$. 

In the proof of Theorem~\ref{t:snk=4}, we will need to determine substrings of larger strings. It will be more clear in the exposition if we establish the following convention.

\begin{defn}
Let $r_{i_1}r_{i_2}\cdots r_{i_m}$ represent a cycle in $P_n$ or $BP_n$. We say that a string $s$ is a \emph{continuous substring} of $r_{i_1}r_{i_2}\cdots r_{i_m}$ if $s$ or its reversal can be written in the form $r'_{j_1}r'_{j_2}\cdots r'_{j_{\ell}}$ where $j_{k+1}=j_k+1\pmod m$ for $1\leq k\leq \ell-1$.
\end{defn}

\begin{figure}
\begin{center}
\begin{tikzpicture}[scale=1,every node/.style={scale=1}]
		\node (e) at (-1.5,3.5) {$1234$};
			
		\node (2) at (-2.5,2.5) {$3214$};
		\node (1) at (-0.5,2.5) {$2134$};
		\node (3) at (1.5,3.5) {$4321$};
			
		\node (32) at (2.5,-2.5) {$4123$};
		\node (12) at (-2.5,1.5) {$2314$};
		\node (21) at (-0.5,1.5) {$3124$};
		\node (31) at (-0.5,-2.5) {$4312$};
		\node (13) at (0.5,2.5) {$3421$};
		\node (23) at (2.5,2.5) {$2341$};

		\node (232) at (1.5,-3.5) {$2143$};						
		\node (132) at (2.5,-1.5) {$1423$};
		\node (312) at (-2.5,-1.5) {$4132$};
		\node (212) at (-1.5,0.5) {$1324$};
		\node (321) at (0.5,-1.5) {$4213$};
		\node (231) at (-0.5,-1.5) {$1342$};
		\node (131) at (-1.5,-3.5) {$3412$};
		\node (213) at (0.5,1.5) {$2431$};
		\node (313) at (0.5,-2.5) {$1243$};
		\node (123) at (2.5,1.5) {$3241$};
		\node (323) at (-2.5,-2.5) {$1432$};
			
		\node (1321) at (1.5,-0.5) {$2413$};
		\node (1231) at (-1.5,-0.5) {$3142$};
		\node (1213) at (1.5,0.5) {$4231$};
			
		\draw[line width=2.pt,red](e)--(1) (2)--(12) (3)--(13) (31)--(131) (21)--(212) (32)--(132) (23)--(123) (231)--(1231) (321)--(1321) (213)--(1213);
		\draw[line width=2.pt,blue] (e)--(2) (1)--(21) (3)--(23) (31)--(231) (12)--(212) (32)--(232) (13)--(213) (312)--(1231) (132)--(1321) (123)--(1213);
		\draw[line width=2.pt,purple!50!black] (e)--(3) (2)--(32) (21)--(321) (12)--(312) (23)--(323) (212)--(1213);% (1) arc (100:260:2 and 2.5) (313) arc (-80:80:2 and 2.5); %(1)--(31) (13)--(313);
		\draw[line width=2.pt,purple!50!black,domain=-72:72] plot ({0.25+2*cos(\x)},{2.5*sin(\x)}) plot ({-0.25+2*cos(\x+180)},{2.5*sin(\x+180)});
		
		\draw[line width=2.pt] (232) edge[red] (313) 
		          (232) edge[purple!50!black] (131) 
		          (132) edge[purple!50!black] (123) 
		          (321) edge[blue] (313) 
		          (231) edge[purple!50!black] (213)
		          (312) edge[red] (323)
		          (131) edge[blue] (323)
		          (1321) edge[purple!50!black] (1231);	
		          
		\draw[line width=2.pt] (e) edge[preaction={draw,yellow,-,double=yellow,double distance=2\pgflinewidth,},purple!50!black] (3)
		          (3) edge[preaction={draw,yellow,-,double=yellow,double distance=2\pgflinewidth,},blue] (23)
		          (23) edge[preaction={draw,yellow,-,double=yellow,double distance=2\pgflinewidth,},red] (123)
		          (123) edge[preaction={draw,yellow,-,double=yellow,double distance=2\pgflinewidth,},purple!50!black] (132)
		          (132) edge[preaction={draw,yellow,-,double=yellow,double distance=2\pgflinewidth,},red] (32)
		          (32) edge[preaction={draw,yellow,-,double=yellow,double distance=2\pgflinewidth,},blue] (232)
		          (232) edge[preaction={draw,yellow,-,double=yellow,double distance=2\pgflinewidth,},purple!50!black] (131)
		          (131) edge[preaction={draw,yellow,-,double=yellow,double distance=2\pgflinewidth,},blue] (323)
		          (323) edge[preaction={draw,yellow,-,double=yellow,double distance=2\pgflinewidth,},red] (312)
		          (312) edge[preaction={draw,yellow,-,double=yellow,double distance=2\pgflinewidth,},purple!50!black] (12)
		          (12) edge[preaction={draw,yellow,-,double=yellow,double distance=2\pgflinewidth,},red] (2)
		          (2) edge[preaction={draw,yellow,-,double=yellow,double distance=2\pgflinewidth,},blue] (e);
	\end{tikzpicture}	
\caption{Pancake graph $P_4$. The different colors indicate the different pancake generators. The 12-cycle $r_4r_3r_2r_4r_2r_3r_4r_3r_2r_4r_2r_3$ is highlighted.}
\label{f:network}
\end{center}
\end{figure}
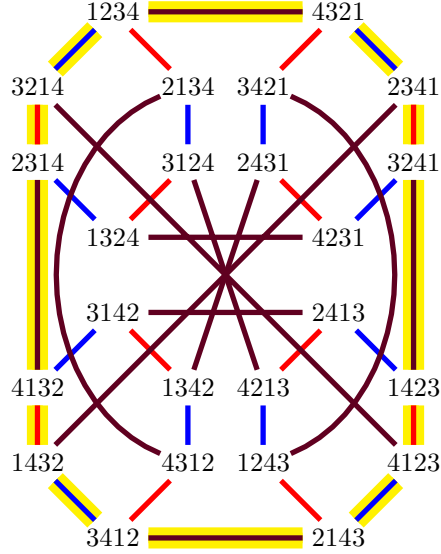

For example, consider the cycle $r_4r_3r_2r_4r_2r_3r_4r_3r_2r_4r_2r_3$ (highlighted in Figure~\ref{f:network}.) Then $r_2r_4r_2r_3r_4$ is a continuous substring as is $r_2r_3r_4r_3$.

Our intention with this definition is to emphasize that if $r_{i_1}r_{i_2}\cdots r_{i_m}$ is a cycle, then $r_{i_m}$ and $r_{i_1}$ should be considered consecutive edges.

\begin{observation}\label{ob:S1} 

 \begin{enumerate}
 \item[(a)] There are six continuous substrings of the form $r_n r_a r_b r_c$ in the only canonical form for a 7-cycle containing $r_n$, $r_n r_{n-1} r_n r_{n-1} r_{n-2} r_n r_2$. These are indicated by an arrow on top of the longer string (the directions indicate how to read the substring): 
 \begin{align*}
 \overrightarrow{r_n r_{n-1} r_n r_{n-1}} r_{n-2} r_n r_2,\qquad& r_n r_{n-1} \overrightarrow{r_n r_{n-1} r_{n-2} r_n}r_2,\\ 
 \overrightarrow{r_n r_{n-1}} r_n r_{n-1} r_{n-2} \overline{r_n r_2},\qquad& \overline{r_n r_{n-1} r_n }r_{n-1} r_{n-2} r_n \overleftarrow{r_2},\\
 r_n r_{n-1} \overleftarrow{ r_n r_{n-1} r_{n-2} r_n} r_2,\text{ and  }& \overline{r_n} r_{n-1} r_n r_{n-1}  \overleftarrow{r_{n-2} r_n r_2}.
 \end{align*}
 We let $Q_0$ denote the set of these continuous substrings. That is,
 \begin{align*}
 Q_0 := \{
 &r_n r_{n-1} r_n r_{n-1},
 r_n r_{n-1} r_{n-2} r_n,\\ 
 &r_{n} r_2 r_n r_{n-1},
 r_n r_{n-1} r_n r_2, \\
 &r_n r_{n-2}r_{n-1} r_n,
 r_{n} r_2 r_n r_{n-2} 
 %, r_n r_2 r_3 r_2
 \}.    
 \end{align*}
 
 \item[(b)] There are six continuous substrings of the form $r_a r_n r_b r_c$ in $r_n r_{n-1} r_n r_{n-1} r_{n-2} r_n r_2$, the only canonical form for a 7-cycle containing $r_n$. These are indicated by an arrow on top of the longer string (the directions indicate how to read the substring): 
 \begin{align*}
 r_n \overrightarrow{r_{n-1} r_n r_{n-1} r_{n-2}} r_n r_2,\qquad& \overrightarrow{r_n} r_{n-1} r_n r_{n-1} \overline{r_{n-2} r_n r_2},\\ 
 \overrightarrow{r_n r_{n-1} r_n} r_{n-1} r_{n-2} r_n \overline{r_2},\qquad& \overleftarrow{r_n r_{n-1} r_n r_{n-1}} r_{n-2} r_n r_2,\\ 
 r_n r_{n-1} r_n \overleftarrow{r_{n-1} r_{n-2} r_n r_2},\text{ and  }& \overline{r_n r_{n-1}} r_n r_{n-1} r_{n-2} \overleftarrow{r_n r_2}.
 \end{align*}
 We let $Q_1$ denote the set of these continuous substrings. That is,
 
 \begin{align*}
 Q_1:= \{
 &r_{n-1} r_n r_{n-1} r_{n-2}, 
 r_{n-2} r_n r_2 r_n, \\
 &r_2 r_n r_{n-1} r_n,
 r_{n-1} r_n r_{n-1} r_n, \\ 
 &r_2 r_n r_{n-2} r_{n-1}, 
 r_{n-1} r_n r_2 r_n
 \}.    
 \end{align*}

 \item[(c)]  There are two continuous substrings of the form $r_a r_b r_n r_c$ with $a \neq n$ in $r_n r_{n-1} r_n r_{n-1} r_{n-2} r_n r_2$. These are shown by an arrow over the top of the canonical form of the 7-cycles (the directions indicate how to read the substring): $r_n r_{n-1} r_n \overrightarrow{r_{n-1} r_{n-2} r_n r_2}$ and $r_n \overleftarrow{r_{n-1} r_n r_{n-1} r_{n-2}} r_n r_2$. We let $Q_2$ denote the set of these continuous substrings. That is,%Therefore, $r_{n-1}r_{n-2}r_nr_2=r_nr_{n-1}r_n$ and $r_{n-1}r_{n-2}r_nr_2=r_nr_2r_n$.
\[Q_2 := \{r_{n-1} r_{n-2} r_n r_2, r_{n-2} r_{n-1} r_n r_{n-1}\}.\]

\item[(d)] Since there is only once 6-cycle, $r_2r_3r_2r_3r_2r_3$, then $r_2r_3r_2r_n=r_3r_2r_3r_n$. We define \begin{displaymath}Q_3 :=\{r_2 r_3 r_2 r_n\}.\end{displaymath} 
 \end{enumerate}
\end{observation}

The sets $Q_0, Q_1, Q_2,$ and $Q_3$ are important to identify since they are those length 4 paths in $P_n$ for which there is a ``short cut,'' of length 3, in $Q_0, Q_1,$ and $Q_2$. While $Q_3$ represents a length 4 path that potentially may be double counted since it has another representation. Thus these paths will need to be accounted for in the cardinalities of the sets $A_i$.

\begin{thm}\label{t:snk=4} If $n\geq4$, then
\[
R_4(n)=\frac{1}{2} (2 n^4 - 15 n^3 + 29 n^2 + 6 n - 34).
\]
\end{thm}

\begin{proof} We will use PIE with the sets $A_i$, for $0\leq i\leq 4$ and (\ref{eq:PIE}). We verify each of the cardinalities, one by one. To obtain the relevant cardinality of intersections of the sets, we will use Theorem~\ref{t:cycleclassification} since some permutations that can be sorted in four flips can also be sorted in three, and even if said permutations require four flips to be sorted, there might be more than one way of doing so. We explain each of the cases in detail in the proof. To obtain a recurrence for $R_k(n)$, we will assume that $n>4$. 
\begin{description}
\item[$|A_4|$.] Due to the recursive structure of $P_n$, it follows that $|A_4| = R_4(n-1)$.
\item[$|A_3|$.] Any permutation in $A_3$ will be of the form $r_a r_b r_c r_n$ where $2 \leq a,b,c \leq n-1, a \neq b,$ and $b\neq c$. We can choose $a$ in $n-2$ ways and $b,c$ in $n-3$ ways. By Observation~\ref{ob:S1}(d), $r_2r_3r_2r_n=r_3r_2r_3r_n$ so to avoid double counting we remove the string in $Q_3$ from the set of strings of the form $r_a r_b r_c r_n$ to get $A_3$. Hence, $|A_3|=(n-2)(n-3)^2-1$.
\item[$|A_2|$.] Any permutation in $A_2$ will be of the form $r_a r_b r_n r_c$ where $2 \leq a,b,c \leq n-1$ and $a\neq b$. We can choose $a,c$ in $n-2$ ways each and $b$ in $n-3$ ways, so there are $(n-2)^2(n-3)$ strings of this form. All of these strings represent permutations that can be sorted in four flips, but some of these can be also be sorted in three flips. So we need to exclude those permutations that can be sorted in three flips to get $A_2$. Notice that any permutation that can be sorted in four flips and also in three flips will be part of a 7-cycle in $P_n$. By Observation~\ref{ob:S1}(c), there are two strings of this form, those in $Q_2$, that lead to the same permutation. Therefore, $|A_2| = (n-2)^2(n-3) - 2$.
\item[$|A_1|$.] Any permutation in $A_1$ will be of the form $r_a r_n r_b r_c$ where $2 \leq a,b \leq n-1$, $2 \leq c \leq n$, and $b\neq c$. We can choose $a, b, c$ in $n-2$ ways each, and so there are $(n-2)^3$ strings of this form. All of these strings represent permutations that can be sorted in four flips but some of these can be also be sorted in three flips. Notice that any permutation that can be sorted in four flips and also in three flips will be part of a 7-cycle in $P_n$. So if we can find $r_a r_n r_b r_c$ as a continuous substring of $r_n r_{n-1} r_n r_{n-1} r_{n-2} r_n r_2$ (the only canonical form for a 7-cycle containing $r_n$), then the permutation obtained from $r_a r_n r_b r_c$ can be sorted in three flips, and should not be counted in $A_1$. By Observation~\ref{ob:S1}(b), there are six strings that should not be counted in $A_1$, those in $Q_1$. Furthermore, unlike in the previous case, there are some 8-cycles to consider: The two strings of the form $r_a r_n r_b r_c$ and $r_z r_n r_y r_x$ (with $2\leq x\leq n, 2\leq y,z\leq n-1$) would represent the same permutation, if they are part of an 8-cycle of the form $r_a r_n r_b r_c r_x r_y r_n r_z$. Comparing with the canonical forms of the 8-cycles in Theorem~\ref{t:cycleclassification}. First, there are $2(n-4)$ cycles of the form~(\ref{P8-5}),
\begin{align*}
&r_{i-1} r_n r_{n-i+1} r_{n-i} \ r_n r_i r_n r_{n-1}, \quad \text{for } 3 \leq i \leq n-2, \text{ and}\\
&r_{n-i} r_n r_{i} r_{n} \ r_{n-1} r_{i-1} r_n r_{n-i+1}, \quad \text{for } 3 \leq i \leq n-2.
\end{align*}
Second, there are $n-4$ cycles of the form~(\ref{P8-6}): $r_{n-i-1} r_n r_{i} r_{i+1} \ r_{n} r_{n-1} r_n r_{n-i}$  with $2 \leq i \leq n-3$. For a total of $3(n-4)$ such strings that are double counted of the form $r_ar_nr_br_c$. Therefore, $|A_1| = (n-2)^3-3(n-4)-6$.
\item[$|A_0|$.] Any permutation in $A_0$ will be of the form $r_n r_a r_b r_c$ where $2 \leq a \leq n-1$, $2 \leq b,c \leq n$ and $b\neq c$. We can choose $a, b, c$ in $n-2$ ways each, so there are $(n-2)^3$ strings that follow this form. The two strings $r_n r_2 r_3 r_2 $ and $r_n r_3 r_2 r_3 $ lead to the same permutation because of the only 6-cycle $r_3r_2r_3r_2r_3r_2$ in $P_n$ (so $r_3r_2r_3=r_2r_3r_2$) so we will exclude the string $r_n r_2 r_3 r_2$ to avoid double counting. As before, we need to account for strings that yield permutations that are part of 7-cycles, since these can be sorted in three flips.  Let $\pi=r_n r_a r_b r_c$ denote a permutation that can be sorted in three flips, then $\pi$ must form part of a 7-cycle. %with the identity permutation $e$. 
By Observation~\ref{ob:S1}(a), there are six strings, those in $Q_0$, that should not be counted as part of $A_0$. In this case, no 8-cycle leads to double counting as if there exists $r_n r_x r_y r_z $ with $2\leq x\leq n-1$ and $2\leq y,z\leq n$ with $r_nr_ar_br_c=r_nr_xr_yr_z$, then $r_n r_a r_b r_c r_z r_y r_x r_n$ must be an 8-cycle, which is not possible by observing the canonical forms for 8-cycles in $P_n$ given in Theorem~\ref{t:cycleclassification}. Therefore, $|A_0| = (n-2)^3 -7$.

\item[$|A_0 \cap A_1|$.] If $\pi\in A_0\cap A_1$, then there must be an 8-cycle of the form $r_x r_n r_y r_z r_a r_b r_c r_n$ with $2 \leq a,b,z \leq n$, $2 \leq c,x,y \leq n-1$, $y \neq z$, $a \neq z$, $a \neq b$, and $c \neq b$ where $r_x r_n r_y r_z \not\in Q_1$ and $r_n r_c r_b r_a \notin Q_0$. Since in this cycle the two $r_n$s are separated by one reversal, no such 8-cycles exist of the form (\ref{P8-1}), (\ref{P8-2}), nor (\ref{P8-3}). 

By comparing with (\ref{P8-4}), we see that the form
\begin{align*}
&r_{n-i+1} r_n r_i r_k \ r_{n-i} r_{n-1} r_{i-1} r_n, 
\end{align*}
with $ 3 \leq i \leq n-2$, contributes $n-4$ to $|A_0 \cap A_1|$.

By comparing with (\ref{P8-5}) we see that the forms
\begin{align*}
&r_i r_n r_{n-1} r_{i-1} \ r_n r_{n-i+1} r_{n-i} r_n, \quad \text{for } 3 \leq i \leq n-2,\text{ and }\\
&r_{i} r_n r_{n-i} r_{n-i+1} \ r_n r_{i-1} r_{n-1} r_n, \quad \text{for } 3 \leq i \leq n-2,
\end{align*}
contribute $2(n-4)$ to $|A_0 \cap A_1|$.

By comparing with (\ref{P8-6}) we see that the form
\begin{align*}
&r_{n-1} r_n r_{n-i} r_{n-i-1} \ r_n r_i r_{i+1} r_n,
\end{align*}
with $2 \leq i \leq n-3$, contributes $n-4$ to $|A_0 \cap A_1|$.

By comparing with (\ref{P8-7}) we get that the following cycles
\begin{align*}
&r_{n-j+1} r_n r_i r_n \ r_{n-j+1} r_n r_i r_n,
\end{align*}
with $2 \leq i < j\leq n-1$, contribute $\frac{1}{2}(n-2)(n-3)$ to $|A_0 \cap A_1|$. However, when $i=2$ and $j=3$ the cycle is $r_{n-2}r_nr_2r_nr_{n-2}r_nr_2r_n$. So $r_{n-2}r_nr_2r_n=r_nr_2r_nr_{n-2}$, but the left-hand-side string is in $Q_1$ and the right-hand-side string is in $Q_0$. Thus this cycle should be excluded, since both strings actually represent the same permutation $r_{n-1}r_nr_{n-1}$. Therefore, $|A_0 \cap A_1| = \frac{1}{2}(n-2)(n-3) + 4(n-4) -1 = \frac{1}{2}(n^2 + 3n - 28)$.

\item [$|A_0 \cap A_2|$.] If $\pi \in A_0 \cap A_2$, then there must be an 8-cycle of the form $r_x r_y r_n r_z r_a r_b r_c r_n$ with $2 \leq a,b,x \leq n$, $2 \leq c,y,z \leq n-1$, $x \neq y$, $z \neq a$, $a \neq b$, and $b \neq c$ where $r_x r_y r_n r_z \notin Q_2$ and $r_n r_c r_b r_a \notin Q_0$. 
Since in this cycle a pair of $r_n$s are separated by two reversals, no such cycle of the form (\ref{P8-1}), (\ref{P8-2}), (\ref{P8-3}), (\ref{P8-4}), nor (\ref{P8-7}) exists.

Upon comparison with (\ref{P8-5}) we see that the forms
\begin{align*}
&r_{n-1}r_{i-1}r_{n}r_{n-i+1}\ r_{n-i}r_{n}r_{i}r_{n}, \quad \text{for } 3 \leq i \leq n-2,\\
&r_{n-i+1}r_{n-i}r_{n}r_{i}\ r_{n}r_{n-1}r_{i-1}r_{n}, \quad \text{for } 3 \leq i \leq n-2,\\
&r_{n-i}r_{n-i+1}r_{n}r_{i-1}\ r_{n-1}r_{n}r_{i}r_{n}, \quad \text{for } 3 \leq i \leq n-2, \text{ and }\\
&r_{i-1}r_{n-1}r_{n}r_{i}\ r_{n}r_{n-i}r_{n-i+1}r_{n}, \quad \text{for } 3 \leq i \leq n-2,\\
\end{align*}
contribute $4(n-4)$ to $|A_0 \cap A_2|$.

Furthermore, upon comparing with (\ref{P8-6}), we see that
\begin{align*}
&r_{n-i}r_{n-i-1}r_{n}r_{i}\ r_{i+1}r_{n}r_{n-1}r_{n}, \quad \text{for } 2 \leq i \leq n-3, \text{ and }\\
&r_{i}r_{i+1}r_{n}r_{n-1}\ r_{n}r_{n-i}r_{n-i-1}r_{n}, \quad \text{for } 2 \leq i \leq n-3,
\end{align*}
contribute $2(n-4)$ to $|A_0 \cap A_2|$. 

Hence $|A_0 \cap A_2| = 6(n-4)$.

\item[$|A_0 \cap A_3|$.] If $\pi\in A_0\cap A_3$, then there must be an 8-cycle of the form $r_x r_y r_z r_n r_a r_b r_c r_n$ with $2 \leq a,c,x,y,z \leq n-1$, $2 \leq b \leq n$, $x \neq y$, $y \neq z$, $a \neq b$, and $b \neq c$ where $r_x r_y r_z r_n \notin Q_1$ and $r_n r_c r_b r_a \notin Q_0$. 
Since in this cycle a pair of $r_n$s are separated by three reversals, no such cycle of the form (\ref{P8-5}), (\ref{P8-6}), nor (\ref{P8-7}) exists.

Let us compare with form (\ref{P8-1}), we see that the form
\begin{align*}
    &r_{j}r_{i}r_{j}r_{n}\ r_{n-j+i}r_{i}r_{n-j+i}r_{n},
\end{align*} 
with $2 \leq i <j \leq n-1$, contributes $\frac{1}{2}(n-2)(n-3)$ to $|A_0\cap A_3|$. 

When we compare with form (\ref{P8-2}),
%\begin{align*}
    $r_{2}r_{3}r_{2}r_{n}\ r_{n-1}r_{n}r_{i}r_{n}$
%\end{align*}
actually does not contribute to $|A_0\cap A_3|$ since $r_{2}r_{3}r_{2}r_{n}\in Q_3$.

We now compare with (\ref{P8-3}),
\begin{align*}
    r_{n-i}r_{n-1}r_{i}r_{n}r_{n-i}r_{n-1}r_{i}r_{n},
\end{align*}
with $ 2 \leq i \leq n-2$. Each value of $i$ gives a valid cycle (neither of the paths are in $Q_0 \ nor \ Q_3$).

We now compare with (\ref{P8-4}),
\begin{align*}
    &r_{n-i}r_{n-1}r_{i-1}r_{n}r_{n-i+1}r_{n}r_{i}r_{n}, \quad \text{for } 3 \leq i \leq n-2.
\end{align*}
This form contributes $n-4$ to $|A_0 \cap A_3|$. 

Hence, \begin{displaymath}|A_0 \cap A_3| = \frac{1}{2}(n-2)(n-3) + (n-4) +(n-3) = \frac{1}{2}(n^2 - n - 8). \end{displaymath}
%%\frac{1}{2}(n-2)(n-3) + (n-4) +(n-3-\chi_{n=4}) = \frac{1}{2}(n^2 - n - 8) - \chi_{n=4}.\end{displaymath}% \hfill{where $I=1 \ for \ n=4 \ else \ I=0$}

\item[$|A_0 \cap A_4|$.] If a permutation $\pi \in A_0 \cap A_4$, then there must be an 8-cycle of the form $r_w r_x r_y r_z r_a r_b r_c r_n$, where $2 \leq w,x,y,z \leq n-1$ and $r_n r_c r_b r_a \notin Q_0$. No such 8-cycle exist since %this cycle has $5$ points in that copy of $P_{n-1}$ containing the identity permutation, and so $|A_0 \cap A_4|=0$
there are no cycles with four reversals between a pair of $r_{n}$ in Theorem \ref{t:cycleclassification}. So $|A_0 \cap A_4|=0$.

\item[$|A_1 \cap A_2|$.] If a permutation $\pi \in A_1 \cap A_2$, then there must be an 8-cycle of the form $r_x r_y r_n r_z r_a r_b r_n r_c$ with $2 \leq a \leq n$, $2 \leq b,c,x,y,z \leq n-1$, $x \neq y$, $z \neq a$, and $a \neq b$ where $r_x r_y r_n r_z \notin Q_2$ and $r_c r_n r_b r_a \notin Q_1$. %We now use the canonical form of the 8-cycles.
Since in this cycle a pair of $r_{n}$ are separated by three reversals, no such cycle of the form (\ref{P8-5}), (\ref{P8-6}), nor (\ref{P8-7}) exists.

From form (\ref{P8-1}), we get: $r_{i}r_{j}r_{n}r_{n-j+i} \ r_{i}r_{n-j+i}r_{n}r_{j}$ with $2 \leq i < j \leq n-1$. Each value of $i$ and $j$ give a valid cycle (the paths are not in $Q_1 \ or \ Q_2$) except when $i=n-2$ and $j=n-1$. Hence the contribution of (\ref{P8-1}) is $\frac{1}{2}(n-2)(n-3)-1$. 

From form (\ref{P8-2}), we get: $r_{2}r_{n-1}r_{n}r_{2}r_{3}r_{2}r_{n}r_{n-1}$ and
$r_{3}r_{2}r_{n}r_{n-1}r_{2}r_{n-1}r_{n}r_{2}$. Thus this form contributes $2$ to $|A_1 \cap A_2|$. 

From form (\ref{P8-3}), we get: $r_{n-1}r_{i}r_{n}r_{n-i} \ r_{n-1}r_{i}r_{n}r_{n-i}$ with $2 \leq i \leq n-2$. Each value of $i$ except for $i=n-2$ gives a valid cycle (the paths are not in $Q_1 \ or \ Q_2$). Thus this form contributes $n-4$ to $|A_1 \cap A_2|$. 

From form (\ref{P8-4}) we get: $r_{n-1}r_{i-1}r_{n}r_{n-i+1} \ r_{n}r_{i}r_{n}r_{n-i}$ with $3 \leq i \leq n-2$. If we replace $i-1$ with $i$ in this form, we obtain the cycles whose first four reversals are the same as those obtained from the form (\ref{P8-3}) in the previous paragraph. So the permutations that were contributed from form (\ref{P8-4}) were already considered and would be double counted. Thus this particular form contributes nothing to $|A_1 \cap A_2|$.

Therefore, $|A_1 \cap A_2| = \frac{1}{2}(n-2)(n-3) + (n-4)+1 = \frac{1}{2}(n^2-3n)$.

\item[Other intersections of two sets.] All other intersections $A_1 \cap A_3,A_1 \cap A_4,A_2 \cap A_3,A_2 \cap A_4$, and $A_3 \cap A_4$ are empty. Otherwise, there would be an 8-cycle that could not be matched to any of the canonical forms of Theorem~\ref{t:cycleclassification}. Specifically, it can seen that the number of reversals between any two $r_n$ would be greater than or equal to four, which does not occur in any of the 8-cycles in Theorem~\ref{t:cycleclassification}.

\item[$|A_0 \cap A_1 \cap A_2|$.] If a permutation belongs to these three sets, $\pi\in A_0\cap A_1\cap A_2$, then it must be that the permutation can be written in the forms $\pi=r_n r_w r_v r_u$, $\pi=r_c r_n r_b r_a$ and $\pi=r_x r_y r_n r_z$ with $2 \leq b,c,x,y,w,z \leq n-1$ and $2 \leq a,u,v \leq n$ and none in $Q_0, Q_1,$ or $Q_2$. Thus, $P_n$ must contain 8-cycles of the following forms:
\begin{align}
    \label{3int1} &r_x r_y r_n r_z r_a r_b r_n r_c, \\
    \label{3int2} &r_x r_y r_n r_z r_u r_v r_w r_n, \text{ or}\\
    \label{3int3} &r_c r_n r_a r_b r_u r_v r_w r_n.
\end{align}

Notice that if there are cycles of the form (\ref{3int1}) and of the form (\ref{3int2}), there will be cycles that can be written in the form (\ref{3int3}) as well. Once again we will use the classification of the 8-cycles given in Theorem~\ref{t:cycleclassification} to see if there are cycles that can be written in the form (\ref{3int1}), that would share the first four reversals from (\ref{3int2}) at the same time. 

Comparing (\ref{P8-1}) with (\ref{3int1}) we obtain the cycles of the form \begin{equation}\label{eq:3int-1}
    r_i r_j r_n r_{n-j+i} \ r_i r_{n-j+i}r_n r_j,
\end{equation} with $2 \leq i < j \leq n-1$. 
With the substitutions of $n-1$ for $j$ and $i-1$ for $i$, these cycles become
\begin{equation*}
    r_{i-1} r_{n-1} r_{n} r_{i} \ r_{i-1} r_{i} r_{n} r_{n-1},
\end{equation*} with $3 \leq i \leq n-1$. Notice when comparing with the form (\ref{P8-5}),
\begin{equation*}
    r_{i-1}r_{n-1}r_n r_i r_n r_{n-i} r_{n-i+1} r_n, 
\end{equation*} with $2 \leq i \leq n-3$ the first four reversals match. So each value of $i$ in the overlap of indices, $3 \leq i \leq n-2$, corresponds to a permutation in this intersection and contributes a total of $n-4$ to $|A_0 \cap A_1 \cap A_2|$.

Now substituting $i+1$ for $j$ in (\ref{eq:3int-1}), we obtain the following form 
\begin{equation*}
r_{i} r_{i+1} r_{n} r_{n-1} r_{i} r_{n-1} r_{n} r_{i+1},
\end{equation*}
which matches (\ref{P8-6})
\begin{equation*}
r_i r_{i+1} r_n r_{n-1} r_n r_{n-i} r_{n-i-1} r_n,    
\end{equation*}
in the first four reversals, with $2 \leq i \leq n-3$. So each value $i$ in the above form contributes a permutation to the intersection for a total of $n-4$ to $|A_0 \cap A_1 \cap A_2|$. No other canonical 8-cycles would have the first four reversals from the form (\ref{3int1}), and so the total contribution to $|A_0 \cap A_1 \cap A_2|$ starting from (\ref{P8-1}) is $2(n-4)$.

Moreover, comparing (\ref{P8-2}) with (\ref{3int1}), we do not find any additional cycles that are also of the form (\ref{3int2}). 

Furthermore, comparing (\ref{P8-3}) with (\ref{3int1}), we find the form
\begin{equation}\label{eq:3int-4}
 r_{n-1}r_{i}r_{n}r_{n-i} \ r_{n-1}r_{i}r_{n}r_{n-i},   
\end{equation} with $2 \leq i \leq n-2$. By substituting $i-1$ for $i$ in (\ref{eq:3int-4}), we obtain 
\begin{equation*}
 r_{n-1}r_{i-1}r_{n}r_{n-i+1} \ r_{n-1}r_{i-1}r_{n}r_{n-i+1},
\end{equation*} with $3 \leq i \leq n-1$. The first four reversals of which match with form (\ref{P8-5})
\begin{equation*}
r_{n-1}r_{i-1}r_{n}r_{n-i+1}\ r_{n-i}r_{n}r_{i}r_{n},    
\end{equation*} with $3 \leq i \leq n-2$. So each value of $i$ in the overlap of these intervals, $3 \leq i \leq n-2$, contribute a permutation to this intersection for a total of $n-4$ added to $|A_0 \cap A_1 \cap A_2|$. No other matches with form (\ref{P8-3}) are obtained, and thus the contribution to $|A_0 \cap A_1 \cap A_2|$ from (\ref{P8-3}) is $n-4$.

Finally, we compare (\ref{P8-4}) with (\ref{3int1}) and obtain the form
\begin{equation}\label{eq:3int-6}
r_{n-1}r_{i-1}r_{n}r_{n-i+1} \ r_{n}r_{i}r_{n}r_{n-i},
\end{equation} with $3 \leq i \leq n-2$. If we compare form (\ref{eq:3int-6}) with any of the other canonical forms of 8-cycles in $P_n$, we would find a match with form (\ref{P8-3}) after substituting $i-1$ for $i$ in (\ref{eq:3int-6}). However, these cycles have already been counted. 

This completes the count of permutations found in $A_0 \cap A_1 \cap A_2$, and thus $|A_0 \cap A_1 \cap A_2|=3(n-4)$.
\end{description}

An exhaustive argument gives that any other intersections of three of the sets $A_0,\ldots,A_4$ will be empty, since no cycle matching three forms at the same time exists. Therefore if $n>4$,
\begin{align*}
    R_4(n) &= \left|\bigcup_{i=0}^4 A_i\right| \\
    &= \sum_{S \subseteq \{0,1,2,3,4\}, S \neq \emptyset} (-1)^{|S|+1} \bigcap_{i \in S} A_i\\ 
          &= |A_0| + |A_1| + |A_2| + |A_3| + |A_4| - |A_0 \cap A_1| -
            |A_0 \cap A_2| - |A_0 \cap A_3| - |A_1 \cap A_2|\\
            & \quad + |A_0 \cap A_1 \cap A_2|\\
            &=R_4(n-1) + \frac{1}{2}(8n^3 - 57n^2 + 111n - 40).
\end{align*} Using the initial condition $R_4(4)=3$ (see Table~\ref{tab:sn}) and solving the recurrence, it follows that if $n\geq4$,
\[
R_4(n)=\frac{1}{2} (2 n^4 - 15 n^3 + 29 n^2 + 6 n - 34),
\] as desired. 
\end{proof}

% Table

%\begin{landscape}
\begin{table}%\label{tab:sn}
\centering
\tiny
\scalebox{1.1}{
\begin{tabular}{|*{13}{c|}}
   \hline
    \diagbox{$n$}{$k$} & 0 & 1 & 2 & 3 & 4 & 5 & 6 & 7 & 8 & 9 & 10 & 11 \\
    \hline
    1 & 1 & 0 & 0 & 0 & 0 & 0 & 0 & 0 & 0 & 0 & 0 & 0\\
    \hline
    2 & 1 & 1 & 0 & 0 & 0 & 0 & 0 & 0 & 0 & 0 & 0 & 0\\
    \hline
    3 & 1 & 2 & 2 & 1 & 0 & 0 & 0 & 0 & 0 & 0 & 0 & 0\\
    \hline
    4 & 1 & 3 & 6 & 11 & 3 & 0 & 0 & 0 & 0 & 0 & 0 & 0\\
    \hline
    5 & 1 & 4 & 12 & 35 & 48 & 20 & 0 & 0 & 0 & 0 & 0 & 0\\
    \hline
    6 & 1 & 5 & 20 & 79 & 199 & 281 & 133 & 2 & 0 & 0 & 0 & 0\\
    \hline
    7 & 1 & 6 & 30 & 149 & 543 & 1357 & 1903 & 1016 & 35 & 0 & 0 & 0\\
    \hline
    8 & 1 & 7 & 42 & 251 & 1191 & 4281 & 10561 & 15011 & 8520 & 455 & 0 & 0\\
    \hline
    9 & 1 & 8 & 56 & 391 & 2278 & 10666 & 38015 & 93585 & 132697 & 79379 & 5804 & 0\\
    \hline
    10 & 1 & 9 & 72 & 575 & 3963 & 22825 & 106461 & 377863 & 919365 & 1309756 & 814678 & 73232\\
    \hline
    11 & 1 & 10 & 90 & 809 & 6429 & 43891 & 252737 & 1174766 & 4126515 & 9981073 & 14250471 & 9123648\\
    \hline
    12 & 1 & 11 & 110 & 1099 & 9883 & 77937 & 533397 & 3064788 & 14141929 & 49337252 & 118420043 & 169332213\\
    \hline
    13 & 1 & 12 & 132 & 1451 & 14556 & 130096 & 1030505 & 7046318 & 40309555 & 184992275 & 639783475 & 1525125357\\
    \hline
    14 & 1 & 13 & 156 & 1871 & 20703 & 206681 & 1858149 & 14721545 &  100464346 & 572626637 &&\\
    \hline
    15 & 1 & 14 & 182 & 2365 & 28603 & 315305 & 3169675 & 28528986 & 226016576 &&&\\
    \hline
    16 & 1 & 15 & 210 & 2939 & 38559 & 465001 & 5165641 & 52027677 & 468966948 &&&\\
    \hline
    17 & 1 & 16 & 240 & 3599 & 50898 & 666342 & 8102491 & 90238067 & 911274131 &&&\\
    \hline
    18 & 1 & 17 & 272 & 4351 & 65971 & 931561 & 12301949 & 150044655 & 1677036683 &&&\\
    \hline
    19 & 1 & 18 & 306 & 5201 & 84153 & 1274671 & 18161133 & 240665410 & 2947991637 &&&\\
    \hline
    20 & 1 & 19 & 342 & 6155 & 105843 & 1711585 & 26163389 &374193014 & 4982872347 &&&\\
    \hline
    21 & 1 & 20 & 380 & 7219 & 131464 & 2260236 & 36889845 &566212968 & 8141208511 &&&\\
    \hline
    % 22 & 1 & 21 &&&&&&&&&&\\
    % \hline
    % 23 & 1 & 22 &&&&&&&&&&\\
    % \hline
    % 24 & 1 & 23 &&&&&&&&&&\\
    % \hline
    % 25 & 1 & 24 &&&&&&&&&&\\
    % \hline
    % 26 & 1 & 25 &&&&&&&&&&\\
    % \hline
    % 27 & 1 & 26 &&&&&&&&&&\\
    % \hline
    % 28 & 1 & 27 &&&&&&&&&&\\
    %\hline
\end{tabular}
}
\caption{Numbers of the form $R_k(n)$ for several values of $n$ and $k$. In particular, notice that if $n\geq4$, $R_4(n)=\frac{1}{2} (2 n^4 - 15 n^3 + 29 n^2 + 6 n - 34)$. The empty entries are unknown to us.}
\label{tab:sn}
\end{table}
%\end{landscape}

\section{Signed permutations requiring four flips to be sorted}\label{sec:bn}

We write $R_k^B(n)$ to denote the number of sign permutations that require $k$ burnt pancake flips to be sorted. Since the smallest cycle that can be found in $BP_n$ has length 8 (See~\cite{Compeau2011}), it follows that if $n\geq1$,
    $R_1^B(n)=n,R_2^B(n)=n(n-1)$, and $R_3^B(n)=n(n-1)^2.$
    
    So the first non-trivial case is the computation of $R_4^B(n)$. We will follow the same method used in Section~\ref{sec:sn}: We will define certain sets whose union will the set of all signed permutations requiring four flips to sort. The cardinality of this union will then equal the number of burnt pancake stacks that require four flips to be sorted. The computation of the cardinality of the union of the sets that we will define is carried out utilizing the principle of inclusion-exclusion. We will use the classification of canonical forms of the 8-cycles, due to the authors~\cite{BBP2018}, from Theorem~\ref{t:burntpancakeclassification}.

Our aim is to obtain the cardinality of the set $BA^4:=\{\pi\in B_n: d^B(e,\pi)=4\}$. We furthermore let $BA_i\subseteq BA^4$, $0\leq i\leq 4$, be the following sets.

\begin{enumerate}
    \item If $0\leq i\leq 3$,  \begin{displaymath}BA_i := \{\pi=r^B_{j_1}r^B_{j_2}r^B_{j_3}r^B_{j_4} \in BA^4: j_{i+1}=n,j_k\neq n\text{ for }k<i+1\}.\end{displaymath}
    \item If $i=4$, \begin{displaymath}BA_4 := \{\pi=r^B_{j_1}r^B_{j_2}r^B_{j_3}r^B_{j_4} \in BA^4:  j_k\neq n \text{ for all }k\in[4]\}.\end{displaymath}
\end{enumerate}

In other words, if one were to think of a path between the identity and $\pi\in BA^4$ in $BP_n$, $BA_i$ would contain all the paths that have $i+1$ vertices inside $BP_{n-1}(n)$. It follows that 
\begin{equation}\label{eq:bPIE}
R^B_4(n)=|BA^4|=\left|\bigcup_{i=0}^n BA_i\right|.
\end{equation}
Just like in the previous section, we will use PIE to compute the cardinality of the union $\bigcup_{i=0}^n BA_i$. 

In the proof of the main theorem of this section, the following two results will be used.

\begin{lem}[Lemma 4.5 in~\cite{BBP2018}]\label{l:4.5}
If $\pi_1, \pi_2 \in V (BP_{n-1}(p))$, for any $p \in [\pm n]$, with $d^B(\pi_1, \pi_2) \leq 2$, then $\pi_1r^B_n$ and $\pi_2r^B_n$ must belong to distinct copies of $BP_{n-1}$ in $BP_n$.
\end{lem} Moreover, the following corollary also follows. 

\begin{cor}\label{c:4vertices}
Let $C$ be an 8-cycle in $BP_n$, with $n\geq2$. If $C$ has vertices in exactly two copies $BP_{n-1}(i)$ and $BP_{n-1}(j)$ with $i,j\in[\pm n]$, then $C$ has four vertices in $BP_{n-1}(i)$ and four vertices in $BP_{n-1}(j)$.
\end{cor}
\begin{proof}
By Lemma \ref{l:4.5}, if the endpoints in the $BP_{k-1}(p)$ copy (say $\pi_1$ and $\pi_2$) are at a distance of at most two, then $\pi_1r_k$ and $\pi_2r_k$ will belong to distinct copies of $BP_{k-1}$. Hence an 8-cycle cannot occur in such a way that it has six vertices in one copy of $BP_{n-1}$ and two in the other, or with five vertices in one copy and two in the other. Therefore an 8-cycle with vertices in exactly two copies of $BP_{n-1}$ can only have four vertices in each of the copies.
\end{proof}

We now state and prove the main theorem of this section, that is, an explicit formula for $R^B_4(n)$.

\begin{thm}\label{t:bsnk=4} If $n\geq1$, then
\[
    R_4^B(n) = \frac{1}{2} n(n-1)^2(2n-3).
\]
\end{thm}
\begin{proof} We will use PIE with the sets $BA_i$, $0\leq i\leq 4$, and (\ref{eq:bPIE}), and the canonical forms for the 8-cycles from Theorem~\ref{t:burntpancakeclassification}. We analyze each of the cardinalities individually. To derive a recurrence for $R^B_k(n)$, we will first assume that $n>3$.
\begin{description}
\item[$|BA_4|$.] Due to the recursive structure of $BP_n$, it follows that $|BA_4| = R_4(n-1)$.
\item[$|BA_3|$.] Any permutation in $BA_3$ will be of the form $r^B_ar^B_br^B_cr^B_n$ where $1 \leq a,b,c \leq n-1$, $a\neq b$, and $b\neq c$. We can choose $a$ in $n-1$ ways and $b, \ c$ in $n-2$ ways each. Since there are no 6-, nor 7-cycles in $BP_n$, each choice of $a,b,c$ will give a different signed permutation that requires four flips to be sorted. Indeed, if there were two strings $r^B_dr^B_er^B_fr^B_n$ and $r^B_xr^B_by^B_cz^B_n$ corresponding to the same permutation, $r^B_dr^B_er^B_fr^B_zr^B_yr^B_x$ would be a 6-cycle in $BP_n$, and these do not exist. Therefore, $|BA_3|=(n-1)(n-2)^2$.
\item[$|BA_2|$.] Any permutation in $BA_2$ will be of the form $r^B_ar^B_br^B_nr^B_c$ where $1 \leq a,b,c \leq n-1$, $a\neq b$. All of these signed permutations require four flips to be sorted, as there are no 6-, nor 7-cycles. Furthermore, if two strings of this form produced the same signed permutation, then there would be an 8-cycle of the form $r^B_d r^B_e r^B_n r^B_f r^B_z r^B_n r^B_y r^B_x$, which cannot be placed in any of the canonical forms given in Theorem~\ref{t:burntpancakeclassification}. Therefore, $|BA_2|=(n-1)^2(n-2)$.
\item[$|BA_1|$.] Any permutation in $BA_1$ will be of the form $r^B_a r^B_n r^B_b r^B_c$ where $1 \leq a,b \leq n-1$, $1 \leq c \leq n$, and $b\neq c$. We can choose $a, b, c$ in $n-1$. The same arguments presented in the previous cases yield that each of these strings give a different signed permutation, and so $|BA_1|=(n-1)^3$.  
\item[$|BA_0|$.] Any permutation in $BA_0$ will be of the form $r^B_n r^B_a r^B_b r^B_c$ where $1\leq a\leq n-1$, $1 \leq b,c \leq n$, and $b\neq c$. Since there are no 6-, nor 7- cycles, each of these strings will lead to a different signed permutation, and so $|BA_0|=(n-1)^3$. 
\item[$|BA_0 \cap BA_1|$.] If $\pi \in BA_0 \cap BA_1$, then there must be an 8-cycle of the form $r^B_n r^B_a r^B_b r^B_c r^B_z r^B_y r^B_n r^B_x$ where $1 \leq a,x,y \leq n-1$, $1 \leq b,c,y,z \leq n$, $a \neq b$, $b \neq c$, and $y \neq z$. We find contributions to this form from the canonical forms (\ref{BP-8-2}), (\ref{BP-8-3}), and (\ref{BP-8-4}) only. For form (\ref{BP-8-2}), we obtain the following cycles
\begin{equation}\label{eq:b1}
  r^B_{n}r^B_{j}r^B_{n}r^B_{i}r^B_{n}r^B_{j}r^B_{n}r^B_{i},
\end{equation} with $2 \leq i,j \leq n-2, i+j \leq n$. Considering the possible values of $i,j$ can take in (\ref{eq:b1}), we get $\frac{1}{2}(n-3)(n-2)$ 8-cycles.

Similarly, comparing with (\ref{BP-8-3}) we get $2(n-2)$ 8-cycles arising from all possible values of $i$ in the forms below.
\begin{align*}
    r^B_n r^B_i r^B_n r^B_1 r^B_n r^B_i r^B_n r^B_1, &\text{ with }2 \leq i \leq n-1, \text{ and}\\
    r^B_n r^B_1 r^B_n r^B_i r^B_n r^B_1 r^B_n r^B_i, &\text{ with }2 \leq i \leq n-1.
\end{align*}

Furthermore, by comparing with (\ref{BP-8-4}) we get only one 8-cycle: $r^B_n r^B_1 r^B_n r^B_1 r^B_n r^B_1 r^B_n r^B_1 $. Putting the pieces together, we have 
\begin{displaymath}|BA_0 \cap BA_1| = \frac{1}{2}(n-2)(n-3)+2(n-2)+1.\end{displaymath}

\item[$|BA_0 \cap BA_2|$.] If $\pi \in BA_0 \cap BA_2$, then there must be an 8-cycle of the form $r^B_n r^B_a r^B_b r^B_c r^B_z r^B_n r^B_y r^B_x$. If such a cycle existed, it would have five vertices in one copy of $BP_{n-1}$, which contradicts Corollary~\ref{c:4vertices}. Hence, $|BA_0 \cap BA_2|=0$.

\item[$|BA_0 \cap BA_3|$.] If $\pi \in BA_0 \cap BA_3$, then there must be an 8-cycle of the form $r^B_n r^B_a r^B_b r^B_c r^B_n r^B_z r^B_y r_x$ where $1 \leq a,c,x,y,z \leq n-1$, $1 \leq b \leq n$, $a \neq b$, $b \neq c$, $x \neq y$, and $y \neq z$. The only canonical form that can
match this is (\ref{BP-8-1}), obtaining 
\begin{equation}\label{eq:b2}
r^B_{n}r^B_{j}r^B_{i}r^B_{j}r^B_{n}r^B_{n-j+i}r^B_{i}r^B_{n-j+i},
\end{equation}
with $1\leq i<j \leq n-1$. There are $\frac{1}{2}(n-1)(n-2)$ possible values for $i,j$ in (\ref{eq:b2}), and so \begin{displaymath}|BA_0 \cap BA_3| = \frac{1}{2}(n-1)(n-2).\end{displaymath}

\item[$|BA_0 \cap BA_4|$.] Essentially the same argument as the case $BA_0 \cap BA_2$ gives that $|BA_0 \cap BA_4|=0$.

\item[$|BA_1 \cap BA_2|$.] If $\pi \in BA_1 \cap BA_2$, then there must be an 8-cycle of the form $r^B_a r^B_n r^B_b r^B_c r^B_z r^B_n r^B_y r^B_x$ where $1 \leq a,b,x,y,z \leq n-1$, $1 \leq c \leq n$, $a \neq x$, $b \neq c$, $c \neq z$, and $x \neq y$. The only canonical form from Theorem~\ref{t:burntpancakeclassification} that matches this form is (\ref{BP-8-1}), we get 
\begin{equation}\label{eq:b3}
r^B_{n-j+i}r^B_{n}r^B_{j}r^B_{i}r^B_{j}r^B_{n}r^B_{n-j+i}r^B_{i},
\end{equation}  with $1\leq i<j \leq n-1$. Considering all possible values for $i,j$ in (\ref{eq:b3}), we have 

\begin{displaymath}|BA_1 \cap BA_2| = \frac{1}{2}(n-1)(n-2).\end{displaymath}

\item[Other intersections] All other intersections $BA_1 \cap BA_3, BA_1 \cap BA_4,BA_2 \cap BA_3, BA_2 \cap BA_4$, and $BA_3 \cap BA_4$ are empty. By the same token, all the intersections of three distinct sets from $\{BA_i\}_{i=0}^4$ are empty as well. Indeed, if one of these intersections were not empty, then there would be an 8-cycle that could not be matched to any of the canonical forms of Theorem~\ref{t:burntpancakeclassification}.
\end{description} Now, using PIE, if $n>3$ it follows that
\begin{align*}
    R^B_4(n) &= \left|\bigcup_{i=0}^4BA_i\right| \\
          &= \sum_{S \subseteq \{0,1,2,3,4\}, S \neq \emptyset} (-1)^{|S|+1} \bigcap_{i \in S} BA_i\\ 
          &= |BA_0| + |BA_1| + |BA_2| + |BA_3| + |BA_4| - |BA_0 \cap BA_1| -|BA_0 \cap BA_3|\\ &\quad - |BA_1 \cap BA_2|\\
          &=R^B_4(n-1) +  \frac{1}{2}(8n^3-33n^2+45n-20).
\end{align*} After solving the recurrence relation, using the initial condition $R^B_4(3)=18$ (see Table~\ref{tab:bn}), we obtain that for $n\geq3$,
\[
R^B_4(n)=\frac{1}{2} n(n-1)^2(2n-3).
\] Upon further inspection, it turns out that the if we plug in $n=1,2$ into $\frac{1}{2} n(n-1)^2(2n-3)$ we obtain $0,1$ respectively. Since these are indeed the true values of $R^B_4(1)$ and $R^B_4(2)$, we have that for $n\geq1$, $R^B_4(n)=\frac{1}{2} n(n-1)^2(2n-3)$. This completes the proof of the theorem. 
\end{proof}

\begin{table}%\label{tab:bn}
\centering
\tiny
\scalebox{1.1}{
\begin{tabular}{|*{13}{c|}}
   \hline
    \diagbox{$n$}{$k$} & 0 & 1 & 2 & 3 & 4 & 5 & 6 & 7 & 8 & 9 & 10 & 11\\
    \hline
    1 & 1 & 1 & 0 & 0 & 0 & 0 & 0 & 0 & 0 & 0 & 0 & 0\\
    \hline
    2 & 1 & 2 & 2 & 2 & 1 & 0 & 0 & 0 & 0 & 0 & 0 & 0\\
    \hline
    3 & 1 & 3 & 6 & 12 & 18 & 6 & 2 & 0 & 0 & 0 & 0 & 0\\
    \hline
    4 & 1 & 4 & 12 & 36 & 90 & 124 & 96 & 18 & 3 & 0 & 0 & 0\\
    \hline
    5 & 1 & 5 & 20 & 80 & 280 & 680 & 1214 & 1127 & 389 & 40 & 4 & 0\\
    \hline
    6 & 1 & 6 & 30 & 150 & 675 & 2340 & 6604 & 12795 & 15519 & 6957 & 959 & 43\\
    \hline
    7 & 1 & 7 & 42 & 252 & 1386 & 6230 & 24024 & 71568 & 159326 & 222995 & 136301 & 21951\\
    \hline
    8 & 1 & 8 & 56 & 392 & 2548 & 14056 & 68656 & 276136 & 901970 & 2195663 & 3531887 & 2743477\\
    \hline
    9 & 1 & 9 & 72 & 576 & 4320 & 28224 & 166740 & 843822 & 3636954 & 12675375 & 33773653 & 60758618\\
    \hline
    10 & 1 & 10 & 90 & 810 & 6885 & 51960 & 359928 & 2193534 & 11738418 & 53257425 & 198586153 & \\
    \hline
    11 & 1 & 11 & 110 & 1100 & 10450 & 89430 & 710358 & 5060220 & 32328648 & 180577749 &  & \\
    \hline
    12 & 1 & 12 & 132 & 1452 & 15246 & 145860 & 1306448 & 10645866 & 79016157 &  &  & \\
    \hline
    13 & 1 & 13 & 156 & 1872 & 21528 & 227656 & 2269410 & 20812077 & 175905015 &  &  & \\
    \hline
    14 & 1 & 14 & 182 & 2366 & 29575 & 342524 & 3760484 & 38319281 & 363216425 &  &  & \\
    \hline
    15 & 1 & 15 & 210 & 2940 & 39690 & 499590 & 5988892 & 67117596 &  &  &  & \\
    \hline
    16 & 1 & 16 & 240 & 3600 & 52200 & 709520 & 9220512 & 112694400 &  &  &  & \\
    \hline
    17 & 1 & 17 & 272 & 4352 & 67456 & 984640 & 13787272 & 182483644 &  &  &  & \\
    \hline
    18 & 1 & 18 & 306 & 5202 & 85833 & 1339056 & 20097264 & 286341948 &  &  &  & \\
    \hline
    19 & 1 & 19 & 342 & 6156 & 107730 & 1788774 & 28645578 &  &  &  &  & \\
    \hline
    20 & 1 & 20 & 380 & 7220 & 133570 & 2351820 & 40025856 &  &  &  &  & \\
    \hline
    21 & 1 & 21 & 420 & 8400 & 163800 & 3048360 & 54942566 &  &  &  &  & \\
    \hline
    22 & 1 & 22 & 462 & 9702 & 198891 & 3900820 & 74223996 &  &  &  &  & \\
    \hline
    23 & 1 & 23 & 506 & 11132 & 239338 & 4934006 & 98835968 &  &  &  &  & \\
    \hline
    24 & 1 & 24 & 552 & 12696 & 285660 & 6175224 & 129896272 &  &  &  &  & \\
    \hline
    25 & 1 & 25 & 600 & 14400 & 338400 & 7654400 & 168689820 &  &  &  &  & \\
    \hline
\end{tabular}
}
\caption{Numbers of the form $R_k^B(n)$ for several values of $n$ and $k$. In particular, notice that if $n\geq1$, then $R_4^B(n)=\frac{1}{2} n(n-1)^2(2n-3)$. The empty entries are unknown to us.}
\label{tab:bn}
\end{table}

\section{Classification of the 9-cycles in the Burnt Pancake Graph}\label{sec:9cycles}

In this section, we present classification of any 9-cycle in $BP_n$, with $n\geq2$. This presentation is in the same spirit as \cite{BBP2018,KM10,KM11,KonMed,KM16} where similar forms for 6-,7-,8- and 9-cycles in the pancake graphs $P_n$ and 8-cycles in the burnt pancake graph $BP_n$ are given. We start the description of $9$-cycles in $BP_n$, with $n\geq2$, by giving some preliminary definitions, notation, and lemmas. 

In classifying the 9-cycles we will look at decomposing the window notation of the a signed permutation $\sigma \in B_n$ into substrings, $\sigma = [X Y Z]$. A convenient notation that will be employed is for a signed reversal of a substring, that is if $X=[x_1\;x_2\;\cdots\;x_{i-1}\;x_i]$, then $\overline{X}=[\uline{x_i}\;\uline{x_{i-1}}\;\cdots\;\uline{x_2}\;\uline{x_1}]$. As is customary, for a graph $G$ we shall use $V(G)$ for its set of vertices and $E(G)$ for its set of edges. In this section, we generally follow the convention of using names of signed permutations based on the last character, e.g., $\pi \in V(BP_{n-1}(p))$, $\uline{\pi} \in V(BP_{n-1}(\uline{p}))$, $\rho \in V(BP_{n-1}(q))$, etc. 

In addition to Lemma~\ref{l:4.5}, a few other lemmas from \cite{BBP2018} will be necessary in the classification of all the 9-cycles in $BP_n$, with $n\geq3$. We recall that $BP_2$ is itself an 8-cycle~\cite[Theorem 10]{Compeau2011}, so if a 9-cycle exists in $BP_n$, then $n\geq3$.
\begin{lem}[Lemma 4.2 in \cite{BBP2018}]\label{lem:31}
 If $\pi \in V(BP_{n-1}(p))$ and $\pi r^B_{n} \in V(BP_{n-1}(q))$, then $|p| \neq |q|$.
\end{lem}

Moreover, the following lemma is also used. 
\begin{lem}[Lemma 4.3 in \cite{BBP2018}]\label{lem:32}
Let $\pi,\tau\in B_n$ have the same first element $q\in[\pm n]$ in window notation. Then $d(\pi,\tau) = 3$ if and only if $\tau = \pi 
%f_{j-1} f_{i-1} f_{j-1}$
r^B_{j} r^B_{i} r^B_{j}$, $1 \leq i < j \leq n$ where $\pi = [A B C]$, $\tau=[A \overline{B} C]$, $|A|=j-i$, $|B|=i$, and $|C| \geq 0$.
\end{lem}

We are now ready to state and prove the main result of this section, the classification of all the 9-cycles in $BP_n$, with $n\geq3$.

\begin{thm}\label{t:canonical9cyclesB} If $n\geq3$, then the canonical forms of 9-cycles in $BP_n$ are as follows:  
\begin{align}
    \label{9-1}
    &r^B_{k}r^B_{k-i}r^B_{k}r^B_{k-j}r^B_{k-i-j}r^B_{k}r^B_{j}r^B_{i+j}r^B_{i} &\quad 1 \leq i,j \leq k-2, \quad i+j \leq k-1, \quad 3 \leq k \leq n;\\
    \label{9-2}
    &r^B_{k}r^B_{i+j}r^B_{i}r^B_{k}r^B_{k-i}r^B_{j}r^B_{k}r^B_{k-j}r^B_{k-i-j} &\quad 1 \leq i,j \leq k-2, \quad i+j \leq k-1, \quad 3 \leq k \leq n.
\end{align}
\end{thm}

%\noindent \textbf{Proof of Theorem 2}\\[3mm]
\begin{proof}
As mentioned before, $BP_k$ has a recursive structure and we can find $2k$ copies of $BP_{k-1}$ embedded into $BP_k$, with $3\leq k\leq n$. Recall that we use $BP_{k-1}(x)$, with $x\in[\pm k]$, to denote the subgraph isomorphic to $BP_{k-1}$ induced by looking at the vertices of $BP_k$ that end with the string $x\,(k+1)\,(k+2)\,\cdots\,n$. We will make use of this recursive structure to classify all 9-cycles in $BP_{n-1}(n)$, which is the copy including the identity, by considering the vertices of the 9-cycles in different copies of $BP_{k-1}$ embedded in $BP_k$. Due to the vertex transitive nature of $BP_{n}$, if there is a cycle $C$ in $BP_n$, there would be a cycle with the same labels as $C$ that includes the identity, and therefore it is enough to consider the cycles that contain the identity. 

Since each vertex in $BP_{k-1}(x)$ is connected to exactly one other vertex in some $BP_{k-1}(y)$, with $y\in[\pm k]\setminus\{x\}$, any 9-cycle will share at least two vertices with any copy of $BP_{k-1}$. We will identify a 9-cycle with a partition $(a_1+a_2+\cdots+a_m)$ of 9. That is, $a_1+a_2+\cdots+a_m=9$, where $a_i$ indicates the number of vertices in the $i${th} copy of $BP_{k-1}$ the cycle is incident upon. %We can obtain canonical forms of 9–cycles. 
 As noted above, $a_i \geq 2$ for all $i$. Thus a 9-cycle can be formed by using two, three, or four copies of $BP_{k-1}$. Enumerating through each possible partition will exhaust all possible 9-cycles.
%\newpage
\begin{enumerate}
\item[CASE I :-] A cycle incident upon two copies of $BP_{k-1}$.

We know from Lemma \ref{l:4.5} that if two permutations $\pi_1$ and $\pi_2$ belong to the same copy of $BP_{k-1}$ and are at a distance of less than 3, then $\pi_1r^B_k$ and $\pi_2r^B_k$ belong to different copies of $BP_{k-1}$. Hence it is necessary that five vertices are in one copy and four vertices are in the other.

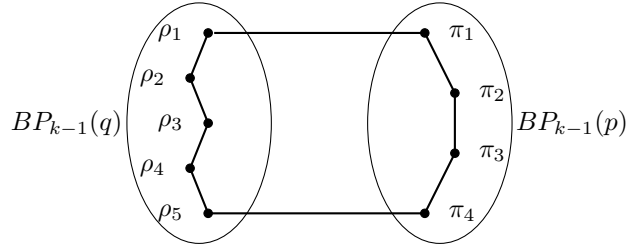
\begin{figure}
\centering
\begin{tikzpicture}[scale=0.8]
\draw (2,2) ellipse (1.2cm and 2cm);
\draw (6,2) ellipse (1.2cm and 2cm);
\draw node at (-0.2,2) {$BP_{k-1}(q)$};
\draw node at (8.2,2) {$BP_{k-1}(p)$};
\filldraw (2.15,3.50) circle (2pt) node[align=left,xshift = -0.5cm] {$\rho_{1}$};
\filldraw (1.85,2.75) circle (2pt) node[align=left,xshift = -0.5cm] {$\rho_{2}$};
\draw [thick] (2.15,3.50) -- (1.85,2.75);
\filldraw (2.15,2.0) circle (2pt) node[align=left,xshift = -0.5cm] {$\rho_{3}$};
\draw [thick] (1.85,2.75) -- (2.15,2);
\filldraw (1.85,1.25) circle (2pt) node[align=left,xshift = -0.5cm] {$\rho_{4}$};
\draw [thick] (2.15,2) -- (1.85,1.25);
\filldraw (2.15,0.5) circle (2pt) node[align=left,xshift = -0.5cm] {$\rho_{5}$};
\draw [thick] (1.85,1.25) -- (2.15,0.5);
\filldraw (5.75,0.50) circle (2pt) node[align=left,xshift = 0.5cm] {$\pi_{4}$};
\draw [thick] (2.15,0.50) -- (5.75,0.50);
\filldraw (6.25,1.50) circle (2pt) node[align=left,xshift = 0.5cm] {$\pi_{3}$};
\draw [thick] (5.75,0.50) --(6.25,1.50);
\filldraw (6.25,2.50) circle (2pt) node[align=left,xshift = 0.5cm] {$\pi_{2}$};
\draw [thick] (6.25,1.50) -- (6.25,2.50);
\filldraw (5.75,3.50) circle (2pt) node[align=left,xshift = 0.5cm] {$\pi_{1}$};
\draw [thick] (6.25,2.50) -- (5.75,3.50);
\draw [thick] (2.25,3.50) -- (5.75,3.50);
\end{tikzpicture}
\caption{A 9-cycle incident on two copies of $BP_{k-1}$ would need to be a (5+4) cycle.}
\label{f:(5+4)}
\end{figure}
%\noindent\\[5mm]

Let the two copies used be $BP_{k-1}(p)$ and $BP_{k-1}(q)$. By Lemma~\ref{lem:31}, $|p| \neq |q|$. So we may track the position and sign of both $p$ and $q$ in a every permutation of the cycle. Suppose that four vertices of such a 9-cycle belong to $BP_{k-1}(p)$, and the other five vertices belong to $BP_{k-1}(q)$ (see Figure~\ref{f:(5+4)}). The four vertices of $BP_{k-1}(p)$ form a path of length three whose endpoints are adjacent to vertices from $BP_{k-1}(q)$, which means both vertices should have $\uline{q}$ in their first positions. Starting with one of these vertices in $BP_{k-1}(p)$ we have the form $\pi_{1}=[\uline{q} X p]$. By Lemma~\ref{lem:32}, we can describe the forms of the remaining vertices of $BP_{k-1}(p)$. With $1 \leq i < j \leq k-1$ we have $\pi_{2}=\pi_{1}r^B_j=[\overline{X_1} q X_2 p]$, $\pi_{3}=\pi_{2}r^B_i=[X_{12} \overline{X_{11}} q X_2 p]$, $\pi_{4}=\pi_{3}r^B_j=[\uline{q} X_{11} \overline{X_{12}} X_2 p]$  where $X=X_1X_2$, $X_1 = X_{11}X_{12}$, $|X_1|=j-1$, and $|X_{12}|=i$. Continuing in $BP_{k-1}(q)$ we see $\rho_{1} = [\uline{p} \overline{X_2} \ \overline{X_{12}} \ \overline{X_{11}} q]$ and $\rho_{5} = [\uline{p} \overline{X_2} X_{12} \overline{X_{11}} q]$. Taking $A = \uline{p}X_2$ and $B=\overline{X_{11}}q$ it is clear that $|A|,|B|,|X_{12}| \geq 1$. We need a path of length four from $\rho_5=[A X_{12} B]$ to $\rho_1=[A \overline{X_{12}} B]$. $r^B_{|A|+|X_{12}|}r^B_{|X_{12}|}r^B_{|A|+|X_{12}|}$ is a path of length three. Thus if a path of length four existed, there would be a 7-cycle in $BP_{k-1}(q)$, which is not possible since the length of the smallest cycle in $BP_n$ is eight. Hence cycles of form (5+4) are not possible.

\item[CASE II :-] A cycle incident upon three copies of $BP_{k-1}$.

There can be three possibilities for the partition of vertices (5+2+2) or (4+3+2) or (3+3+3). Let the three copies incident upon be $BP_{k-1}(p)$, $BP_{k-1}(q)$, and $BP_{k-1}(s)$. By Lemma \ref{lem:31}, it follows that $|p|\neq |q|$, $|p| \neq |s|$, and $|q| \neq |s|$. 

\begin{figure}
\centering
\begin{tikzpicture}[scale=0.8]
\draw (2,2) ellipse (1.2cm and 2cm);
\draw node at (-0.2,2) {$BP_{k-1}(p)$};
\draw [rotate around={55:(5,3.2)}](5,3.2) ellipse (0.75cm and 1.45cm);
\draw node at (7,3.5) {$BP_{k-1}(q)$};
\draw [rotate around={-55:(5,0.8)}](5,0.8) ellipse (0.75cm and 1.45cm);
\draw node at (7,0.5) {$BP_{k-1}(s)$};
\filldraw (2.15,3.50) circle (2pt) node[align=left,xshift = -0.5cm] {$\pi_{1}$};
\filldraw (1.85,2.75) circle (2pt) node[align=left,xshift = -0.5cm] {$\pi_{2}$};
\draw [thick] (2.15,3.50) -- (1.85,2.75);
\filldraw (2.15,2.0) circle (2pt) node[align=left,xshift = -0.5cm] {$\pi_{3}$};
\draw [thick] (1.85,2.75) -- (2.15,2);
\filldraw (1.85,1.25) circle (2pt) node[align=left,xshift = -0.5cm] {$\pi_{4}$};
\draw [thick] (2.15,2) -- (1.85,1.25);
\filldraw (2.15,0.5) circle (2pt) node[align=left,xshift = -0.5cm] {$\pi_{5}$};
\draw [thick] (1.85,1.25) -- (2.15,0.5);
\filldraw (4.6,0.35) circle (2pt) node[align=left,xshift = 0.45cm] {$\sigma_{2}$};
\draw [thick] (2.15,0.50) -- (4.6,0.35);
\filldraw (5.4,1.15) circle (2pt) node[align=left,xshift = 0.4cm] {$\sigma_{1}$};
\draw [thick] (4.6,0.35) -- (5.4,1.15);
\filldraw (5.4,2.85) circle (2pt) node[align=left,xshift = 0.4cm] {$\rho_{2}$};
\draw [thick] (5.4,1.15) -- (5.4,2.85);
\filldraw (4.6,3.65) circle (2pt) node[align=left,xshift = 0.45cm] {$\rho_{1}$};
\draw [thick] (5.4,2.85) -- (4.6,3.65);
\draw [thick] (4.6,3.65) -- (2.15,3.50);
\end{tikzpicture}
\caption{A 9-cycle incident upon three copies of $BP_{k-1}$ with vertex partition (5+2+2).}
\label{f:(5+2+2)}
\end{figure}
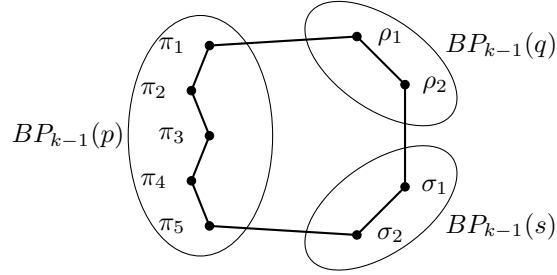

Suppose five vertices of such a 9-cycle belong to the copy $BP_{k-1}(p)$, two vertices belong to a copy $BP_{k-1}(q)$ and the other two vertices belong to a copy $BP_{k-1}(s)$ (see Figure~\ref{f:(5+2+2)}). As $\rho_{2}$ will have $\uline{s}$ in its first position and that $\pi_1$ is exactly two edges away, we see that $\pi_{1}=[\uline{q} X \uline{s} Y p]$. This gives $\rho_{1}=\pi_{1}r^B_k=[\uline{p} \overline{Y} s \overline{X} q]$, $\rho_{2}=\rho_{1}r^B_{|Y|+2}=[\uline{s} Y p \overline{X} q]$, $\sigma_{1}=\rho_{2}r^B_k=[\uline{q} X \uline{p} \overline{Y} s]$. $\sigma_{2}$ must have $\uline{p}$ in its first position, which is not possible in one edge from $\sigma_{1}$. Hence cycles of form (5+2+2) do not exist in the burnt pancake graph.

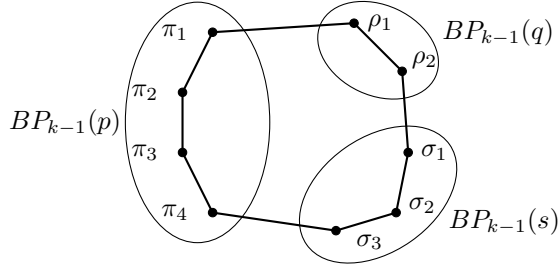
\begin{figure}
\centering
\begin{tikzpicture}[scale=0.8]
\draw (2,2) ellipse (1.2cm and 2cm);
\draw node at (-0.2,2) {$BP_{k-1}(p)$};
\draw [rotate around={65:(5,3.2)}](5,3.2) ellipse (0.75cm and 1.05cm);
\draw node at (7,3.5) {$BP_{k-1}(q)$};
\draw [rotate around={-55:(5,0.8)}](5,0.8) ellipse (0.95cm and 1.45cm);
\draw node at (7.1,0.35) {$BP_{k-1}(s)$};
\filldraw (2.25,3.50) circle (2pt) node[align=left,xshift = -0.5cm] {$\pi_{1}$};
\filldraw (1.75,2.50) circle (2pt) node[align=left,xshift = -0.5cm] {$\pi_{2}$};
\draw [thick] (2.25,3.50) -- (1.75,2.50);
\filldraw (1.75,1.50) circle (2pt) node[align=left,xshift = -0.5cm] {$\pi_{3}$};
\draw [thick] (1.75,2.50) -- (1.75,1.50);
\filldraw (2.25,0.50) circle (2pt) node[align=left,xshift = -0.5cm] {$\pi_{4}$};
\draw [thick] (1.75,1.50) -- (2.25,0.50);
\filldraw (4.3,0.2) circle (2pt) node[align=left,xshift = 0.45cm,yshift=-0.15cm] {$\sigma_{3}$};
\draw [thick] (2.25,0.50) -- (4.3,0.2);
\filldraw (5.3,0.5) circle (2pt) node[align=left,xshift = 0.35cm,yshift = 0.1cm] {$\sigma_{2}$};
\draw [thick] (4.3,0.2) -- (5.3,0.5);
\filldraw (5.5,1.5) circle (2pt) node[align=left,xshift = 0.35cm] {$\sigma_{1}$};
\draw [thick] (5.3,0.5) -- (5.5,1.5);
\filldraw (5.4,2.85) circle (2pt) node[align=left,xshift = 0.3cm,yshift = 0.15cm] {$\rho_{2}$};
\draw [thick] (5.5,1.5) -- (5.4,2.85);
\filldraw (4.6,3.65) circle (2pt) node[align=left,xshift = 0.35cm] {$\rho_{1}$};
\draw [thick] (5.4,2.85) -- (4.6,3.65);
\draw [thick] (4.6,3.65) -- (2.25,3.50);
\end{tikzpicture}
\caption{A 9-cycle incident upon three copies of $BP_{k-1}$ with vertex partition (4+3+2).}
\label{f:(4+3+2)}
\end{figure}

Suppose four vertices of such a 9-cycle belong to the copy $BP_{k-1}(p)$, two vertices belong to a copy $BP_{k-1}(q)$, and the other three vertices belong to a copy $BP_{k-1}(s)$ (see Figure~\ref{f:(4+3+2)}). As $\rho_{2}$ will have $\uline{s}$ in its first position and that $\pi_1$ is exactly two edges away, we see that $\pi_{1}=[\uline{q} X \uline{s} Y p]$. This gives $\rho_{1}=\pi_{1}r^B_k=[\uline{p} \overline{Y} s \overline{X} q]$, $\rho_{2}=\rho_{1}r^B_{|Y|+2}=[\uline{s} Y p \overline{X} q]$, $\sigma_{1}=\rho_{2}r^B_k=[\uline{q} X \uline{p} \overline{Y} s]$. Now, $\sigma_{3}$ must have $\uline{p}$ in its first position, so in the path from $\sigma_{1}$ to $\sigma_{3}$, $p$ should be involved in both the flips.
$\sigma_{2}=\sigma_{1}r^B_{|X|+|Y_2|+2}=[Y_2 p \overline{X} q \overline{Y_1} s]$ where $Y=Y_1Y_2$. This gives $\sigma_{3}=\sigma_{2}r^B_{|Y_2|+1}=[\uline{p} \overline{Y_2}\ \overline{X} q \overline{Y_1} s]$, $\pi_{4}=\sigma_{3}r^B_k=[\uline{s} Y_1 \uline{q} X Y_2 p]$. Taking $A=\uline{s}Y_1$, $B=\uline{q}X$, $C=Y_2 p$ where $|A|,|B|,|C| \geq 1$ we need to find a path of length three between $\pi_4=[ABC]$ and $\pi_1=[BAC]$. One may verify that $\pi_4r^B_{|A|}r^B_{|A|+|B|}r^B_{|B|}=\pi_1$. Taking $|X|=i-1,|Y_1|=j-1$ we get $|Y_2|=k-i-j-1 \geq 0$ and a cycle corresponding to (\ref{9-1}).

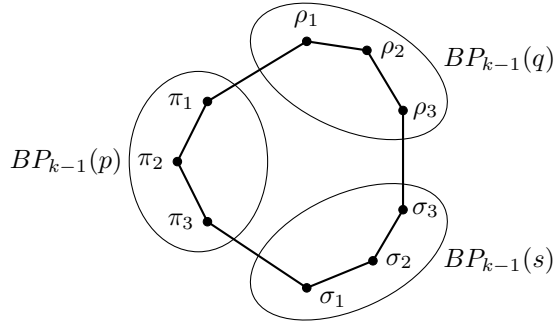
\begin{figure}
\centering
\begin{tikzpicture}[scale=0.8]
\draw (2,2) ellipse (1.15cm and 1.5cm);
\draw node at (-0.2,2) {$BP_{k-1}(p)$};
\draw [rotate around={65:(4.5,3.5)}](4.5,3.5) ellipse (0.95cm and 1.75cm);
\draw node at (7,3.7) {$BP_{k-1}(q)$};
\draw [rotate around={-65:(4.5,0.5)}](4.5,0.5) ellipse (0.95cm and 1.75cm);
\draw node at (7,0.35) {$BP_{k-1}(s)$};
\filldraw (2.15,3.0) circle (2pt) node[align=left,xshift = -0.35cm] {$\pi_{1}$};
\filldraw (1.65,2.0) circle (2pt) node[align=left,xshift = -0.35cm] {$\pi_{2}$};
\draw [thick] (2.15,3) -- (1.65,2);
\filldraw (2.15,1) circle (2pt) node[align=left,xshift = -0.35cm] {$\pi_{3}$};
\draw [thick] (1.65,2) -- (2.15,1);
\filldraw (3.8,-0.1) circle (2pt) node[align=left,xshift = 0.35cm,yshift=-0.15cm] {$\sigma_{1}$};
\draw [thick] (2.15,1) -- (3.8,-0.1);
\filldraw (4.9,0.35) circle (2pt) node[align=left,xshift = 0.35cm] {$\sigma_{2}$};
\draw [thick] (3.8,-0.1) -- (4.9,0.35);
\filldraw (5.4,1.2) circle (2pt) node[align=left,xshift = 0.3cm,yshift=-0.1] {$\sigma_{3}$};
\draw [thick] (4.9,0.35) -- (5.4,1.2);
\filldraw (5.4,2.85) circle (2pt) node[align=left,xshift = 0.3cm] {$\rho_{3}$};
\draw [thick] (5.4,1.2) -- (5.4,2.85);
\filldraw (4.8,3.85) circle (2pt) node[align=left,xshift = 0.35cm] {$\rho_{2}$};
\draw [thick] (5.4,2.85) -- (4.8,3.85);
\filldraw (3.8,4) circle (2pt) node[align=left,yshift = 0.3cm] {$\rho_{1}$};
\draw [thick] (4.8,3.85) -- (3.8,4);
\draw [thick] (3.8,4) -- (2.15,3.0);
\end{tikzpicture}
\caption{A 9-cycle incident upon three copies of $BP_{k-1}$ with vertex partition (3+3+3).}
\label{f:(3+3+3)}
\end{figure}

Suppose three vertices of such a 9-cycle belong to a copy $BP_{k-1}(p)$, three vertices belong to a copy $BP_{k-1}(q)$ and the other three vertices belong to a copy $BP_{k-1}(s)$ (see Figure~\ref{f:(3+3+3)}). The vertex in $BP_{k-1}(p)$ that is adjacent to a vertex in $BP_{k-1}(q)$, $\pi_1$, can be of the form (a) $\pi_1=[\uline{q} X s Y p]$ or (b) $\pi_1=[\uline{q} X \uline{s} Y p]$.

\begin{enumerate}
    \item Since $\pi_{1} = [\uline{q} X s Y p]$, then $\rho_{1}=\pi_{1}r^B_k=[\uline{p} \overline{Y} \uline{s} \overline{X} q]$. As $\pi_{3}$ must have $\uline{s}$ in its first position, $s$ should be involved in only one reversal in the path from $\pi_{1}$ to $\pi_{3}$. So the first reversal must not involve $s$. This gives $\pi_{2}=\pi_{1}r^B_{|X_1|+1}=[\overline{X_1} q X_2 s Y p]$ where $X=X_1X_2$, $\pi_{3}=\pi_{2}r^B_{|X|+2}=[\uline{s} \overline{X_2} \uline{q} X_1 Y p]$, and $\sigma_{1}=\pi_{3}r^B_k=[\uline{p} \overline{Y} \ \overline{X_1} q X_2 s]$.
    As $\sigma_{3}$ must have $\uline{q}$ in its first position, $q$ should be involved in one reversal in the path from $\sigma_{1}$ to $\sigma_{3}$. So the first reversal must not involve $q$. This gives two possibilities: 
    
    \begin{enumerate}
        \item $\sigma_{2}=\sigma_{1}r^B_{|Y_2|+1}[Y_2 p \overline{Y_1} \ \overline{X_1} q X_2 s ]$ where $Y=Y_1Y_2$ and $\sigma_{3}=\sigma_{2}r^B_{|Y|+|X_1|+2} = [\uline{q} X_1 Y_1 \uline{p}$ $\overline{Y_2}X_2 s]$. Following through with this possibility we get $\rho_{3}=[\uline{s} \overline{X_2} Y_2 p \overline{Y_1}\ \overline{X_1} q]$. We need a path of length two from $\rho_{3}$ to $\rho_{1}=[\uline{p} \overline{Y_2}\ \overline{Y_1} \uline{s} \overline{X_2}\ \overline{X_1} q]$. As $\rho_{1}$ has $\uline{p}$ in its first position, $p$ should be involved in only the second reversal. In order for this to be so, without having to exchange the positions of $Y_1$, it must be that $|Y_1|=0$. %So the part of permutation after $p$ does not change in this path, so $\rho_{1}$ should have $[Y_1 X_1 q]$ in the end therefore $|Y_1|=0$. 
        Then we get $\rho_{3}r^B_{|X_2|+1}r^B_{|X_2|+|Y_2|+2}=\rho_{1}$. Taking $|X|=i-1$ and $|X_2|=j-1$ we get $|Y_2|=k-i-j-1 \geq 0$ and a cycle corresponding to (\ref{9-2}).
        
        \item $\sigma_{2}=\sigma_{1}r^B_{|Y|+|X_{12}|+1}=[X_{12} Y p \overline{X_{11}} q X_2 s]$ where $X_1=X_{11}X_{12}$ with $|X_{12}|\geq 1$ and $\sigma_{3}=\sigma_{2}r^B_{|Y|+|X_1|+2} = [\uline{q} X_{11} \uline{p} \overline{Y} p \overline{X_{12}} X_2 s]$. Following through with this possibility we get $\rho_{3}=[\uline{s} \overline{X_2} X_{12} Y p  \overline{X_{11}} q]$. We need a path of length two from $\rho_{3}$ to $\rho_{1}=[\uline{p} \overline{Y}\ \uline{s} \overline{X_2} \ \overline{X_{12}}\ \overline{X_{11}}q]$. As $\rho_{1}$ has $\uline{p}$ in its first position, $p$ should be involved in only the second reversal. We have $\rho_{3}r^B_{|X_2|+|X_{12}|+1}r^B_{|X_{12}}|r^B_{|X_2|+|X_{12}|+|Y|+2}=\rho_{1}$. This is a path of length three which can be reduced to length two if and only if $|X_{12}|=0$. Since $|X_{12}| \geq 1$ this possibility does not give any 9-cycle.
    \end{enumerate}
    
    \item If $\pi_{1} = [\uline{q} X \uline{s} Y p]$, then $\rho_{1}=\pi_{1}r^B_k=[\uline{p} \overline{Y} s \overline{X} q]$. As $\pi_{3}$ must have $\uline{s}$ in its first position, $\uline{s}$ should be involved in both the reversals. This gives $\pi_{2}=\pi_{1}r^B_{|X|+|Y_1|+2}=[\overline{Y_1} s \overline{X} q Y_2 p]$ where $Y=Y_1 Y_2$, $\pi_{3}=\pi_{2}r^B_{|Y_1|+1}=[\uline{s} Y_1 \overline{X} q Y_2 p]$, and $\sigma_{1}=\pi_{3}r^B_k=[\uline{p} \overline{Y_2} \uline{q} X \overline{Y_1} s]$. As $\sigma_{3}$ must have $\uline{q}$ in its first position, $q$ should be involved in both the reversals in the path from $\sigma_{1}$ to $\sigma_{3}$. This gives two possibilities:
    
    \begin{enumerate}
        \item $\sigma_{2}=\sigma_{1}r^B_{|X_1|+|Y_2|+2}=[\overline{X_1} q Y_2 p X_2 \overline{Y_1}s]$ where $X=X_1X_2$ and $\sigma_{3}=\sigma_{2}r^B_{|X_1|+1}= [\uline{q} X_1 Y_2 p X_2 \overline{Y_1} s]$. Following through with this possibility we get $\rho_{3}=\sigma_{3}r^B_k=[\uline{s} Y_1 \overline{X_2} \uline{p}$ $\overline{Y_2}$ $\overline{X_1} q]$. We need a path of length two from $\rho_{3}$ to $\rho_{1}=[\uline{p}\overline{Y_2}\ \overline{Y_1} s \overline{X_2}\ \overline{X_1}q]$. As $\rho_{1}$ has $\uline{p}\overline{Y_2}$ in its prefix, $\overline{X_2}$ with the same sign and ordering of characters, but $s$ will have the opposite sign in $\rho_{3}$, then $\uline{p}\overline{Y_2}$ and $\overline{X_2}$ must be part of both reversals but $s$ must only be part of one. %So the permutation before $\uline{p}$ i.e. $[\uline{s}Y_1 \overline{X_2}]$ will get reversed in $\rho_{1}$. So $\rho_{1}$ should have $[X_2\overline{Y_1}s]$ therefore
        This is only possible if $|X_2|=0$. Then $\rho_{3}r^B_{|Y|+2}r^B_{|Y_2|+1}=\rho_{1}$. Taking $|X_1|=j-1$, $|Y_1|=i-1$ we get $|Y_2|=k-i-j-1$ and a cycle corresponding to (\ref{9-2}).
        
        \item $\sigma_{2}=\sigma_{1}r^B_{|X|+|Y_2|+|Y_{12}|+2}=[Y_{12}\overline{X}qY_2p\overline{Y_{11}}s]$ where $Y_1=Y_{11}Y_{12}$ with $|Y_{12}| \geq 1$ and $\sigma_{3}=\sigma_{2}r^B_{|X|+|Y_{12}|+1}= [\uline{q}X \overline{Y_{12}} Y_2 p\overline{Y_{11}}s]$. Following through with this possibility we get $\rho_{3}=\sigma_{3}r^B_k=[\uline{s}Y_{11}\uline{p}\overline{Y_2}Y_{12}\overline{X}q]$. We need a path of length two from $\rho_{3}$ to $\rho_{1}=[\uline{p} \overline{Y_2}\ \overline{Y_{12}}\ \overline{Y_{11}}s\overline{X}q]$. As $\rho_{1}$ has $\uline{p}$ in its first position, $p$ should be involved in both the reversals in this path. However, it can be seen that $\rho_{3}r^B_{|Y_{11}|+|Y_{12}|+|Y_2|+2}r^B_{|Y_{12}|}r^B_{|Y_{12}|+|Y_2|+1}=\rho_{1}$. This is a path of length three within $BP_{k-1}(q)$ which can be reduced to length two if and only if $|Y_{12}|=0$ but by assumption $|Y_{12}|\geq 1$. Therefore this possibility does not yield a 9-cycle. 
    \end{enumerate}
\end{enumerate}
    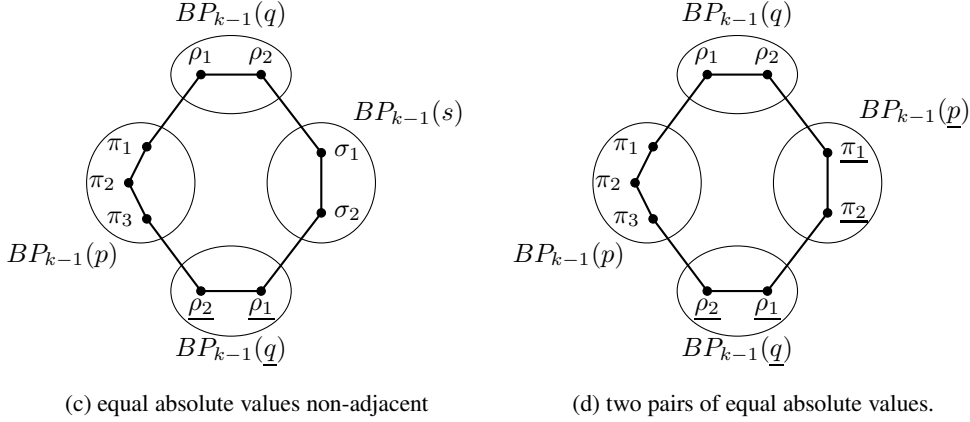
\begin{figure}
        \ContinuedFloat*
        \centering
        \begin{subfigure}{0.45\textwidth}
        \begin{tikzpicture}[scale=0.8]
        \draw (2,2) ellipse (0.9cm and 1cm);
        \draw node at (0.7,0.8) {$BP_{k-1}(p)$};
        \draw (5,2) ellipse (0.9cm and 1cm);
        \draw node at (3.5,4.8) {$BP_{k-1}(q)$};
        \draw (3.5,3.8) ellipse (1cm and 0.65cm);
        \draw node at (6.45,3.15) {$BP_{k-1}(s)$};
        \draw (3.5,0.2) ellipse (1cm and 0.65cm);
        \draw node at (3.5,-0.8) {$BP_{k-1}(t)$};
        \filldraw (2.10,2.6) circle (2pt) node[align=left,xshift = -0.35cm] {$\pi_{1}$};
        \filldraw (1.8,2) circle (2pt) node[align=left,xshift = -0.35cm] {$\pi_{2}$};
        \draw [thick] (2.1,2.6) -- (1.8,2);
        \filldraw (2.1,1.4) circle (2pt) node[align=left,xshift = -0.35cm] {$\pi_{3}$};
        \draw [thick] (1.8,2.0) -- (2.1,1.4);
        \filldraw (3,0.2) circle (2pt) node[align=left,yshift = -0.25cm] {$\tau_{2}$};
        \draw [thick] (2.1,1.4) -- (3,0.20);
        \filldraw (4,0.2) circle (2pt) node[align=left,yshift = -0.25cm] {$\tau_{1}$};
        \draw [thick] (3,0.20) -- (4,0.2);
        \filldraw (5,1.5) circle (2pt) node[align=left,xshift = 0.35cm] {$\sigma_{2}$};
        \draw [thick] (4,0.2) -- (5,1.50);
        \filldraw (5,2.5) circle (2pt) node[align=left,xshift = 0.35cm] {$\sigma_{1}$};
        \draw [thick] (5,1.5) -- (5,2.5);
        \filldraw (4,3.8) circle (2pt) node[align=left,yshift = 0.25cm] {$\rho_{2}$};
        \draw [thick] (5,2.5) -- (4,3.8);
        \filldraw (3,3.8) circle (2pt) node[align=left,yshift = 0.25cm] {$\rho_{1}$};
        \draw [thick] (4,3.8) -- (3,3.8);
        \draw [thick] (3,3.8) -- (2.1,2.6);
        \end{tikzpicture}
        %\caption{A 9-cycle incident upon four copies of $BP_{k-1}$ with vertex partition (3+2+2+2) and pairwise different absolute values.}
        \caption{pairwise different absolute values.}
        \label{f:(3+2+2+2)d}
        \end{subfigure}
        \begin{subfigure}{0.45\textwidth} 
        \begin{tikzpicture}[scale=0.8]
            \draw (2,2) ellipse (0.9cm and 1cm);
            \draw node at (0.7,0.8) {$BP_{k-1}(p)$};
            \draw (5,2) ellipse (0.9cm and 1cm);
            \draw node at (3.5,4.8) {$BP_{k-1}(q)$};
            \draw (3.5,3.8) ellipse (1cm and 0.65cm);
            \draw node at (6.45,3.15) {$BP_{k-1}(\uline{p})$};
            \draw (3.5,0.2) ellipse (1cm and 0.65cm);
            \draw node at (3.5,-0.8) {$BP_{k-1}(t)$};
            \filldraw (2.10,2.6) circle (2pt) node[align=left,xshift = -0.35cm] {$\pi_{1}$};
            \filldraw (1.8,2) circle (2pt) node[align=left,xshift = -0.35cm] {$\pi_{2}$};
            \draw [thick] (2.1,2.6) -- (1.8,2);
            \filldraw (2.1,1.4) circle (2pt) node[align=left,xshift = -0.35cm] {$\pi_{3}$};
            \draw [thick] (1.8,2.0) -- (2.1,1.4);
            \filldraw (3,0.2) circle (2pt) node[align=left,yshift = -0.25cm] {$\tau_{2}$};
            \draw [thick] (2.1,1.4) -- (3,0.20);
            \filldraw (4,0.2) circle (2pt) node[align=left,yshift = -0.25cm] {$\tau_{1}$};
            \draw [thick] (3,0.20) -- (4,0.2);
            \filldraw (5,1.5) circle (2pt) node[align=left,xshift = 0.3cm] {$\uline{\pi_{2}}$};
            \draw [thick] (4,0.2) -- (5,1.50);
            \filldraw (5,2.5) circle (2pt) node[align=left,xshift = 0.3cm] {$\uline{\pi_{1}}$};
            \draw [thick] (5,1.5) -- (5,2.5);
            \filldraw (4,3.8) circle (2pt) node[align=left,yshift = 0.25cm] {$\rho_{2}$};
            \draw [thick] (5,2.5) -- (4,3.8);
            \filldraw (3,3.8) circle (2pt) node[align=left,yshift = 0.25cm] {$\rho_{1}$};
            \draw [thick] (4,3.8) -- (3,3.8);
            \draw [thick] (3,3.8) -- (2.1,2.6);
        \end{tikzpicture}
        %\caption{The copies with equal absolute values are in the non-adjacent copies where one has three vertices.}
        \caption{equal absolute values non-adjacent.}
        \label{f:(3+2+2+2)3-2}
        \end{subfigure}\\
        \caption{9-cycles incident upon four copies of $BP_{k-1}$ with vertex partition (3+2+2+2).}
        \end{figure}
\item[CASE III :-] A cycle incident upon four copies of $BP_{k-1}$.

Due to the constraints on the part sizes in our partition, there can be only one possibility (3+2+2+2).

Let the four copies used be $BP_{k-1}(p)$, $BP_{k-1}(q)$, $BP_{k-1}(s)$, and $BP_{k-1}(t)$ with three vertices in $BP_{k-1}(p)$ and two vertices in each of the other copies. Let us assume, without loss of generality, that one vertex of $BP_{k-1}(p)$ is adjacent to a vertex of $BP_{k-1}(q)$. Here the absolute values of $p,q,s,$ and $t$ may not be distinct. By Lemma \ref{lem:31} only non-adjacent copies can have the same absolute value. This gives rise to three subcases, which we describe below.

\begin{enumerate}
    \item The absolute values of $p,q,s$ and $t$ are pairwise different (see Figure~\ref{f:(3+2+2+2)d}).

        \begin{figure}
        \ContinuedFloat
        \begin{subfigure}{0.45\textwidth}
        \begin{tikzpicture}[scale=0.8]
            \draw (2,2) ellipse (0.9cm and 1cm);
            \draw node at (0.7,0.8) {$BP_{k-1}(p)$};
            \draw (5,2) ellipse (0.9cm and 1cm);
            \draw node at (3.5,4.8) {$BP_{k-1}(q)$};
            \draw (3.5,3.8) ellipse (1cm and 0.65cm);
            \draw node at (6.45,3.15) {$BP_{k-1}(s)$};
            \draw (3.5,0.2) ellipse (1cm and 0.75cm);
            \draw node at (3.5,-0.8) {$BP_{k-1}(\uline{q})$};
            \filldraw (2.10,2.6) circle (2pt) node[align=left,xshift = -0.35cm] {$\pi_{1}$};
            \filldraw (1.8,2) circle (2pt) node[align=left,xshift = -0.35cm] {$\pi_{2}$};
            \draw [thick] (2.1,2.6) -- (1.8,2);
            \filldraw (2.1,1.4) circle (2pt) node[align=left,xshift = -0.35cm] {$\pi_{3}$};
            \draw [thick] (1.8,2.0) -- (2.1,1.4);
            \filldraw (3,0.2) circle (2pt) node[align=left,yshift = -0.25cm] {$\uline{\rho_{2}}$};
            \draw [thick] (2.1,1.4) -- (3,0.20);
            \filldraw (4,0.2) circle (2pt) node[align=left,yshift = -0.25cm] {$\uline{\rho_{1}}$};
            \draw [thick] (3,0.20) -- (4,0.2);
            \filldraw (5,1.5) circle (2pt) node[align=left,xshift = 0.35cm] {$\sigma_{2}$};
            \draw [thick] (4,0.2) -- (5,1.50);
            \filldraw (5,2.5) circle (2pt) node[align=left,xshift = 0.35cm] {$\sigma_{1}$};
            \draw [thick] (5,1.5) -- (5,2.5);
            \filldraw (4,3.8) circle (2pt) node[align=left,yshift = 0.25cm] {$\rho_{2}$};
            \draw [thick] (5,2.5) -- (4,3.8);
            \filldraw (3,3.8) circle (2pt) node[align=left,yshift = 0.25cm] {$\rho_{1}$};
            \draw [thick] (4,3.8) -- (3,3.8);
            \draw [thick] (3,3.8) -- (2.1,2.6);
        \end{tikzpicture}
        %\caption{The copies with equal absolute values are in the non-adjacent copies with only two vertices.}
        \caption{equal absolute values non-adjacent}
        \label{f:(3+2+2+2)2-2}
        \end{subfigure}
        \begin{subfigure}{0.45\textwidth}
        \begin{tikzpicture}[scale=0.8]
        \draw (2,2) ellipse (0.9cm and 1cm);
        \draw node at (0.7,0.8) {$BP_{k-1}(p)$};
        \draw (5,2) ellipse (0.9cm and 1cm);
        \draw node at (3.5,4.8) {$BP_{k-1}(q)$};
        \draw (3.5,3.8) ellipse (1cm and 0.65cm);
        \draw node at (6.45,3.15) {$BP_{k-1}(\uline{p})$};
        \draw (3.5,0.2) ellipse (1cm and 0.75cm);
        \draw node at (3.5,-0.8) {$BP_{k-1}(\uline{q})$};
        \filldraw (2.10,2.6) circle (2pt) node[align=left,xshift = -0.35cm] {$\pi_{1}$};
        \filldraw (1.8,2) circle (2pt) node[align=left,xshift = -0.35cm] {$\pi_{2}$};
        \draw [thick] (2.1,2.6) -- (1.8,2);
        \filldraw (2.1,1.4) circle (2pt) node[align=left,xshift = -0.35cm] {$\pi_{3}$};
        \draw [thick] (1.8,2.0) -- (2.1,1.4);
        \filldraw (3,0.2) circle (2pt) node[align=left,yshift = -0.25cm] {$\uline{\rho_{2}}$};
        \draw [thick] (2.1,1.4) -- (3,0.20);
        \filldraw (4,0.2) circle (2pt) node[align=left,yshift = -0.25cm] {$\uline{\rho_{1}}$};
        \draw [thick] (3,0.20) -- (4,0.2);
        \filldraw (5,1.5) circle (2pt) node[align=left,xshift = 0.35cm] {$\uline{\pi_{2}}$};
        \draw [thick] (4,0.2) -- (5,1.50);
        \filldraw (5,2.5) circle (2pt) node[align=left,xshift = 0.35cm] {$\uline{\pi_{1}}$};
        \draw [thick] (5,1.5) -- (5,2.5);
        \filldraw (4,3.8) circle (2pt) node[align=left,yshift = 0.25cm] {$\rho_{2}$};
        \draw [thick] (5,2.5) -- (4,3.8);
        \filldraw (3,3.8) circle (2pt) node[align=left,yshift = 0.25cm] {$\rho_{1}$};
        \draw [thick] (4,3.8) -- (3,3.8);
        \draw [thick] (3,3.8) -- (2.1,2.6);
        \end{tikzpicture}
        %\caption{A 9-cycle incident upon four copies of $BP_{k-1}$ with vertex partition (3+2+2+2) and two pairs of equal absolute values.}
        \caption{two pairs of equal absolute values.}
        \label{f:(3+2+2+2)p}
        \end{subfigure}
        
        \caption{9-cycles incident upon four copies of $BP_{k-1}$ with vertex partition (3+2+2+2).}
    \end{figure}

    Since none are opposites, $p,q,s,$ and $t$, or their opposites, are present in all the signed permutations in the cycle. Then depending upon the relative position and signs of $s$ and $t$ in our first signed permutation four cases arise. Since $\rho_2$ is two reversals away from $\pi_1$ with one reversing all elements and $\rho_2$ must begin with $\uline{s}$ it is necessary that $\uline{s}$ be in $\pi_1$. The four possible cases of $\pi_1$ are $[\uline{q} X \uline{s} Y t Z p]$, $[\uline{q} X \uline{s} Y \uline{t} Z p]$, $[\uline{q} X t Y \uline{s} Z p]$, and $[\uline{q} X \uline{t} Y \uline{s} Z p]$.
    
    \begin{enumerate}
        \item If $\pi_{1}=[\uline{q} X \uline{s} Y t Z p]$, then $\rho_{1}=\pi_{1}r^B_k=[\uline{p} \overline{Z} \uline{t} \overline{Y} s \overline{X} q]$, $\rho_{2}=\rho_{1}r^B_{|Y|+|Z|+3}=[\uline{s} Y t Z p$ $\overline{X} q]$, and $\sigma_{1}=\rho_{2}r^B_k=[\uline{q} X \uline{p} \overline{Z} \uline{t} \overline{Y} s]$. Now $\sigma_{2}$ must have $\uline{t}$ in its first position, which is not possible in one reversal from $\sigma_{1}$. So this case does not yield any 9-cycles.
        
        \item If $\pi_{1}=[\uline{q} X \uline{s} Y \uline{t} Z p]$, then $\rho_{1}=\pi_{1}r^B_k=[\uline{p} \overline{Z} t \overline{Y} s \overline{X} q]$, $\rho_{2}=\rho_{1}r^B_{|Y|+|Z|+3}=[\uline{s} Y \uline{t} Z p$ $\overline{X} q]$, $\sigma_{1}=\rho_{2}r^B_k=[\uline{q} X \uline{p} \overline{Z} t \overline{Y} s]$, $\sigma_{2}=\sigma_{1}r^B_{|X|+|Z|+3}=[\uline{t} Z p \overline{X} q \overline{Y} s]$, and $\tau_{1}=\sigma_{2}r^B_k=[\uline{s} Y  \uline{q} X \uline{p} \overline{Z} t]$. Now $\tau_{2}$ must have $\uline{p}$ in its first position, which is not possible in one reversal from $\tau_{1}$. So this case does not yield any 9-cycles.

        \item If $\pi_{1}=[\uline{q} X t Y \uline{s} Z p]$, then $\rho_{1}=\pi_{1}r^B_k=[\uline{p} \overline{Z} s \overline{Y} \uline{t} \overline{X} q]$, $\rho_{2}=\rho_{1}r^B_{|Z|+2}=[\uline{s} Z p \overline{Y} \uline{t} \overline{X} q]$, $\sigma_{1}=\rho_{2}r^B_k=[\uline{q} X t Y \uline{p} \overline{Z} s]$, $\sigma_{2}=\sigma_{1}r^B_{|X|+2}=[\uline{t} \overline{X} q Y \uline{p} \overline{Z} s]$, $\tau_{1}=\sigma_{2}r^B_k=[\uline{s} Z p \overline{Y} \uline{q} X t]$, $\tau_{2}=\tau_{1}r^B_{|Z|+2}=[\uline{p} \overline{Z} s \overline{Y} \uline{q} X t]$, and $\pi_{3}=\tau_{2}r^B_k=[\uline{t} \overline{X} q Y \uline{s} Z p]$. We need a path of length two from $\pi_{3}$ to $\pi_{1}$, however, $\pi_{3}r^B_{|X|+2}=\pi_{1}$ which is a path of length one. As there are no 3-cycles in the burnt pancake graph, a path of length two between $\pi_{3}$ and $\pi_{1}$ does not exist. Hence this case does not yield any 9-cycles.
        
        \item If $\pi_{1}=[\uline{q} X \uline{t} Y \uline{s} Z p]$, then $\rho_{1}=\pi_{1}r^B_k=[\uline{p} \overline{Z} s \overline{Y} t \overline{X} q]$, $\rho_{2}=\rho_{1}r^B_{|Z|+2}=[\uline{s} Z p \overline{Y} t \overline{X} q]$, and $\sigma_{1}=\rho_{2}r^B_k=[\uline{q} X \uline{t} Y \uline{p} \overline{Z} s]$. Now $\sigma_{2}$ must have $\uline{t}$ in its first position, which is not possible in one reversal from $\sigma_{1}$. So this case does not yield any 9-cycles.
    \end{enumerate}
    
    \item The absolute values of only one pair among $p,q,s,t$ are the same. This gives rise to two cases. One where the pair of copies with opposite signed last elements have only two vertices each. The other where the pair of copies with opposite signed last elements includes the one copy with three vertices.

    \begin{enumerate}
        \item Say that $s=\uline{p}$, $|q| \neq |t|$ (see Figure~\ref{f:(3+2+2+2)3-2}). In this case $\pi_{1}$ can be either $[\uline{q} X t Y p]$ or $[\uline{q} X \uline{t} Y p]$. 
        
        If $\pi_{1}=[\uline{q} X t Y p]$, then $\rho_{1}=\pi_{1}r^B_k=[\uline{p} \overline{Y} \uline{t} \overline{X} q]$, $\rho_{2}=\rho_{1}r^B_1=[p \overline{Y} \uline{t} \overline{X} q]$, $\uline{\pi_{1}}=\rho_{2}r^B_k=[\uline{q} X t Y \uline{p}]$, $\uline{\pi_{2}}=\uline{\pi_{1}}r^B_{|X|+2}=[\uline{t} \overline{X} q Y \uline{p}]$, $\tau_{1}=\uline{\pi_{2}}r^B_k=[p \overline{Y} \uline{q} X t]$, $\tau_{2}=\tau_{1}r^B_1=[\uline{p} \overline{Y} \uline{q} X t]$, and  $\pi_{3}=\tau_{2}r^B_k=[\uline{t} \overline{X} q Y p]$. We need a path of length two from $\pi_{3}$ to $\pi_{1}$ but $\pi_{3}r^B_{|X|+2}=\pi_{1}$ which is a path of length one. As there are no 3-cycles in the burnt pancake graph a path of length two between $\pi_{3}$ and $\pi_{1}$ does not exist. Hence this case does not yield any 9-cycle.
        
        If $\pi_{1}=[\uline{q} X \uline{t} Y p]$, then $\rho_{1}=\pi_{1}r^B_k=[\uline{p} \overline{Y} t \overline{X} q]$, $\rho_{2}=\rho_{1}r^B_1=[p \overline{Y} t \overline{X} q]$, $\uline{\pi_{1}}=\rho_{2}r^B_k=[\uline{q} X \uline{t} Y \uline{p}]$. As $\uline{\pi_{2}}$ must have $\uline{t}$ in its first position, which is not possible in one reversal from $\uline{\pi_{2}}$. So this case does not yield a 9-cycle.

        \item Say that $t=\uline{q}$, $|p| \neq |s|$ (see Figure~\ref{f:(3+2+2+2)2-2}).
        
        As $\rho_{2}$ will have $\uline{s}$ at first position $\pi_{1}=[\uline{q}X \uline{s}Y p]$. This gives $\rho_{1}=\pi_{1}r^B_k=[\uline{p}\overline{Y}s\overline{X}q]$, $\rho_{2}=\rho_{1}r^B_{|Y|+2}=[\uline{s}Y p\overline{X}q]$, $\sigma_{1}=\rho_{2}r^B_k=[\uline{q}X \uline{p} \overline{Y}s]$, $\sigma_{2}=\sigma_{1}r^B_1=[qX \uline{p} \overline{Y}s]$, $\uline{\rho_{1}}=\sigma_{2}r^B_k=[\uline{s}Y p\overline{X}\uline{q}]$, $\uline{\rho_{2}}=\uline{\rho_{1}}r^B_{|Y|+2}=[\uline{p}\overline{Y}s\overline{X}\uline{q}]$, $\pi_{3}=\uline{\rho_{2}}r^B_k=[qX \uline{s}Y p]$. We need a path of length two from $\pi_{3}$ to $\pi_{1}$ but $\pi_{3}r^B_{1}=\pi_{1}$ which is a path of length one. As there are no 3-cycles in burnt pancake graph a path of length two between $\pi_{3}$ and $\pi_{1}$ does not exist. Hence this case does not give any 9-cycle.
    \end{enumerate}
    %\newpage
    \item The absolute values of the two pairs among $p,q,s,t$ are the same, i.e $s=\uline{p}$ and $t=\uline{q}$ (see Figure~\ref{f:(3+2+2+2)p}).

        Let $\pi_{1}=[\uline{q} X p]$. This gives $\rho_{1}=\pi_{1}r^B_k=[\uline{p} \overline{X} q]$, $\rho_{2}=\rho_{1}r^B_1=[p \overline{X} q]$, $\uline{\pi_{1}}=\rho_{2}r^B_k=[\uline{q} X \uline{p}]$, $\uline{\pi_{2}}=\uline{\pi_{1}}r^B_1=[q X \uline{p}]$, $\uline{\rho_{1}}=\uline{\pi_{2}}r^B_k=[p \overline{X} \uline{q}]$, $\uline{\rho_{2}}=\uline{\rho_{1}}r^B_1=[\uline{p} \overline{X} \uline{q}]$, and $\pi_{3}=\uline{\rho_{2}}r^B_k=[q X p]$. We need a path of length two from $\pi_{3}$ to $\pi_{1}$ but $\pi_{3}r^B_{1}=\pi_{1}$, which is a path of length one. As there are no 3-cycles in the burnt pancake graph a path of length two between $\pi_{3}$ and $\pi_{1}$ does not exist. Hence this case does not yield any 9-cycle.
\end{enumerate}

So there are no cycles of the form (3+2+2+2) in the burnt pancake graph.
\end{enumerate}

This finalizes all possible partitions of the vertices, and thus gives all possible 9-cycles in $BP_{k}$.
\end{proof}

\section{Concluding Remarks}\label{s:conclusion}

In the preceding sections, we have provided explicit formulas for the number of pancake and burnt pancake stacks with $n$ pancakes that require four flips to be sorted, utilizing entirely elementary methods using cycle classification of the pancake and burnt pancake graph and the principle of inclusion-exclusion, as well as providing a classification of all 9-cycles in the burnt pancake graph. 

Having a classification of longer cycles in $P_n$ and $BP_n$ might allow us to take a similar approach to what we used in the proof of Theorem~\ref{t:snk=4} and Theorem~\ref{t:bsnk=4}. Currently, there is no published classification of the 10-cycles for either $P_n$ and $BP_n$. The authors have worked out said classification, though the number of canonical forms alone makes it hard to do a systematic, manual approach using PIE. For example, we found 59 canonical forms for 10-cycles in $P_n$ and 8 canonical forms for 10-cycles in $BP_n$. However, the process would be extremely tedious. 

Utilizing structural characteristics of permutations satisfying certain properties, Homberger and Vatter~\cite{HV16} gave an algorithm that can be used to enumerate certain permutations. In particular, they prove that for ``sufficiently large" $n$, the number of permutations of $[n]$ that can be sorted with $k$ prefix reversals is given by a polynomial. More specifically, if one defines
\begin{equation}\label{eq:tilde}
\widetilde{R}_k(n):=\sum_{i=0}^kR_i(n),
\end{equation} then for sufficiently large $n$, $\widetilde{R}_k(n)$ is given by a polynomial. From (\ref{eq:tilde}) it is possible to compute $R_k(n)$ from $\widetilde{R}_k(n)$ if $R_i(n)$ is known for $0\leq i<k$. In this light, and using the output of the Homberger-Vatter algorithm from~\cite{HV16} for $\widetilde{R}_k(n)$ with $k=5,6$, one obtains the following.
\begin{enumerate}
    \item If $n\geq 5$, then 
\begin{displaymath}R_5(n)= \frac{1}{6}\left(6n^5-65n^4+173n^3+296n^2-1724n+1590\right).\end{displaymath}
    \item If $n\geq6$, then
    \begin{displaymath}R_6(n)= \frac{1}{60}\left(60n^6 - 883n^5 + 3140n^4 + 10775 n^3 - 91400 n^2 + 171068n -58020\right).\end{displaymath}
\end{enumerate} These polynomials explain all the nonzero values for the columns $k=5,6$ in Table~\ref{tab:sn}. The situation for $k>6$ is a bit more interesting, as the polynomial obtained by the Homberger-Vatter algorithm does not explain all nonzero values of $R_k(n)$. Indeed, their algorithm produces polynomials that are valid for ``sufficiently large" $n$. For example, the algorithm correctly computes $\widetilde{R}_7(n)$ if $n>7$. Indeed, if $n>7$, 

\begin{align}\label{eq:r7}
R_7(n)&=\widetilde{R}_7(n)-\sum_{i=1}^6R_i(n)\notag\\
&=\frac{1}{240} (240 n^7- 4619  n^6+21881 n^5+ 109275 n^4 -1372445 n^3+\notag\\
&\hspace{5cm} 4476344 n^2-4550196 n-850320).
\end{align}

However, (\ref{eq:r7}) does not account for $R_7(6)=2$ and $R_7(7)=1016$. Therefore, there is no polynomial that would explain all the nonzero values of $R_7(n)$. 

After computing $\widetilde{R}_8(n)$ using the Homberger-Vatter algorithm, which took several days using a system with Dual Xeon CPUs and 256GB of RAM, we found a polynomial that explains most of the nozero entries of the column $k=8$ in Table~\ref{tab:sn}. Namely, if $n>7$,

\begin{align}\label{eq:r8}
R_8(n)&=\widetilde{R}_8(n)-\sum_{i=1}^7R_i(n)\notag\\
&=\frac{1}{5040}(5040 n^8 - 122683 n^7 + 759857 n^6 + 4519067 n^5 -79101715 n^4 +\notag\\
  &\hspace{3cm}  364661948 n^3 - 561161062 n^2 - 267373812 n + 844945920).
\end{align}

Once again, (\ref{eq:r8}) does not account for $R_8(7)=35$, so once again, there is no polynomial that explains all the nonzero values of $R_8(k)$.

The situation for $R^B_k(n)$ does not seem to be as mysterious. While the Homberger-Vatter algorithm does not apply to signed permutations, the numbers $R^B_k(n)$ seem to be given by polynomials. Using a standard polynomial fitting procedure with the data from Table~\ref{tab:bn}, we obtain the following conjecture. 

\begin{conjecture}\label{con:bn} If $n\geq1$, then
\begin{enumerate}
    \item[(i)] $R^B_5(n)=\frac{1}{6}n(n-1)(n-2)(6n^2-17n+3)$,
    \item[(ii)] $R^B_6(n)=\frac{1}{60}n(n-1)(n-2)(60n^3-343n^2+401n+284)$,
    \item[(iii)] $R^B_7(n)=\frac{1}{240}n(n-1)(n-2)(n-3)(240n^3-1499n^2+925n+5104)$,
    \item[(iv)] $R^B_8(n)=\frac{1}{5040}n(n-1)(n-2)(n-3)(5040n^4-52123n^3+113415n^2+314716n-1027242)$, and
    \item[(v)] $R^B_9(n)=\frac{1}{40320}(n-1) (n-2) (n-3) (n-4) (40320 n^5-444061 n^4+644746 n^3+6638777 n^2-18991470 n).$
\end{enumerate}
\end{conjecture}

%It seems that the values $R_k^B(n)$ are all given by polynomials of degree $k$ if $k\leq 9$. It would be interesting to know if this behavior persists for $k>9$.

These polynomials also conveniently explain the zero entries of Table~\ref{tab:bn}, which seems to indicate that $R^B_k(n)$ are better behaved than those of $R_k(n)$.

Another question that might give interesting results is proving any summation identities that $R_k(n)$ or $R^B_k(n)$ satisfy. From their definitions, it is clear that
\[
\sum_{k\geq0}R_k(n)=|S_n|=n!
\] and that
\[
\sum_{k\geq0}R^B_k(n)=|B_n|=2^nn!.
\] Notice that determining the maximum value of $k$, for a given $n$, such that $R_k(n)$ (or $R^B_k(n)$) is not zero is equivalent to the classic pancake problem (or the burnt pancake problem).

A less trivial identity can be found by looking at the explicit formulas that describe the nonzero values of $R_k(n)$ for $0\leq k\leq 6$, namely, we have the following Corollary. 

\begin{cor}\label{cor:sum} 
If $k\leq6$ and $R_k(n-i)>0$ for $1\leq  i\leq k+1$, then
\[
 R_k(n)=\sum_{i=1}^{k+1}(-1)^{i+1}\binom{k+1}{i}R_{k}(n-i).
 \]
 \end{cor}
 
 Since the polynomials (\ref{eq:r7}) and (\ref{eq:r8}) do not cover the values $R_7(6)$, $R_7(7)$, and $R_8(7)$, Corollary~\ref{cor:sum} would once again require $n$ to be ``sufficiently large.''
 %in the form presented here does not seem to hold in general. 
 
 For signed permutations, we have observed the following identity, which would follow from $R^B_k(n)$ being integer-valued polynomials and by using the Gregory-Newton interpolation formula for integer-valued polynomials (for background see~\cite{CC97}).
 
 \begin{conjecture}\label{con:sum}
 If $k,n\geq1$, then
 \begin{displaymath}R_k^B(n) =\sum_{j=1}^k \left(\sum_{i=0}^{k-j} (-1)^i \binom{i+j-1}{i}\binom{n}{i+j}\right)R_k^B(j), \text{ for } k \geq 1.\end{displaymath}
\end{conjecture} 

We have verified Conjecture~\ref{con:sum} for all the values in Table~\ref{tab:bn}. This conjecture would follow from a result similar to Homberger-Vatter giving that the $R^B_k(n)$ are all polynomials and continue to explain the zero and nonzero values. %However, a general proof is still elusive. 

\section{Acknowledgments} 

A. Patidar was supported by the Global Talent Attraction Program (GTAP) administered by the School of Informatics, Computing, and Engineering at Indiana University. Furthermore, the authors thank V. Vatter for pointing out his and C. Homberger's paper~\cite{HV16} and for making their code publicly available. The authors also wish to thank the anonymous referees for their comments that helped improved the presentation of this paper. 

\bibliographystyle{alpha}
\bibliography{pancake}

\end{document}